\documentclass[11pt, 3p]{elsarticle}

\usepackage[T1]{fontenc}
\usepackage[utf8]{inputenc}
\usepackage{amsmath,amssymb,amsthm,mathtools}
\usepackage{booktabs}
\usepackage{array}
\newcolumntype{L}[1]{>{\raggedright\arraybackslash}p{#1}}
\usepackage{enumitem}
\usepackage{xcolor}
\usepackage{tikz}
\usetikzlibrary{arrows.meta}
\definecolor{keptcol}{HTML}{15496B}
\definecolor{disccol}{HTML}{9AA5AD}
\definecolor{accent}{HTML}{B3541E}
\usepackage{algorithm}
\usepackage{algpseudocode}
\usepackage[hidelinks]{hyperref}

\theoremstyle{plain}
\newtheorem{theorem}{Theorem}[section]
\newtheorem{lemma}[theorem]{Lemma}
\newtheorem{proposition}[theorem]{Proposition}
\newtheorem{corollary}[theorem]{Corollary}
\theoremstyle{definition}
\newtheorem{observation}[theorem]{Observation}

\numberwithin{equation}{section}
\numberwithin{figure}{section}
\numberwithin{table}{section}
\numberwithin{algorithm}{section}

\newcommand{\ColKP}{ColKP}
\newcommand{\KP}{KP}

\newcommand{\OPT}{\mathrm{OPT}}
\newcommand{\LB}{\mathrm{LB}}
\newcommand{\UB}{\mathrm{UB}}
\newcommand{\pmax}{p_{\max}}
\newcommand{\zmin}{{z_{\min}}}
\newcommand{\zmax}{{z_{\max}}}
\newcommand{\scalefactor}{{\lambda}}
\newcommand{\ie}{i.e.,}

\DeclareMathOperator*{\argmax}{arg\,max}
\DeclareMathOperator*{\argmin}{arg\,min}

\newcommand{\items}{\mathcal{I}}
\newcommand{\colors}{\mathcal{C}}
\newcommand{\firstDP}{{\tt DP_1}}
\newcommand{\firstDPp}{{\tt DP^p_1}}
\newcommand{\secondDP}{{\tt DP_2}}
\newcommand{\secondDPp}{{\tt DP^p_2}}

\allowdisplaybreaks

\newcommand{\paperversion}{Version of \today}
\newcommand{\papertitle}{\LARGE A tight $1/3$--approximation algorithm and fully polynomial-time approximation schemes for the Colored Knapsack Problem}
\newcommand{\paperaffiliation}{Department of Computer, Control and Management
Engineering Antonio Ruberti, Sapienza University of Rome, Via Ariosto 25,
00185 Rome, Italy}

\begin{document}

\hypersetup{pageanchor=false}
\begin{frontmatter}
\title{\papertitle}
\date{\paperversion}
\author[sap]{Fabio Ciccarelli\corref{cor}}
\ead{f.ciccarelli@uniroma1.it}
\author[sap]{Fabio Furini}
\ead{fabio.furini@uniroma1.it}
\cortext[cor]{Corresponding author.}
\affiliation[sap]{organization={\paperaffiliation}}

\begin{abstract}
The \emph{Colored Knapsack Problem} (\ColKP) generalizes the classical Knapsack Problem by partitioning the items into color classes and requiring the selected items to admit an ordering in which consecutive items have different colors. The problem is weakly $\mathcal{NP}$-hard and admits two pseudo-polynomial dynamic programming (DP) algorithms proposed in the literature. These two DP algorithms have worst-case running times $O(b \, n^4)$ and $O(b^2 \, n^3)$, respectively, where $b$ is the knapsack capacity and $n$ is the number of items.
We develop the first approximation algorithm for the \ColKP. By rounding an optimal basic solution of the linear programming relaxation of its natural integer programming formulation and repairing color feasibility, we obtain a linear-time approximation-algorithm whose worst-case performance ratio is $1/3$.
We then reformulate both DP algorithms so that profit, rather than knapsack capacity, indexes their pseudo-polynomial dimension, and combine them with profit scaling to obtain two fully polynomial-time approximation schemes (FPTASs). The first FPTAS runs in $O(n^5/\varepsilon)$ time for nonnegative profits and in $O(n^6/\varepsilon)$ time for arbitrary integer profits. The second FPTAS runs instead in $O(n^5/\varepsilon^2)$ and $O(n^7/\varepsilon^2)$ time, respectively. The approximation guarantee, along with new structural insights, provides the bounds needed to control the scaled profit range and establish these running times.
\end{abstract}


\end{frontmatter}
\hypersetup{pageanchor=true}

\section{Introduction}
\label{sec:intro}

Let $\items=\{1,2,\dots,n\}$ be a set of $n$ items and let
$b\in\mathbb{Z}_{>0}$ be the capacity of a knapsack. Each item $i\in\items$
has a weight $w_i\in\mathbb{Z}_{>0}$, a profit $p_i\in\mathbb{Z}$, and a color $\kappa_i\in\colors=\{1,2,\dots,m\}$. 
For a subset $S\subseteq\items$,
we denote its total weight and profit by
\[
    w(S)=\sum_{i\in S}w_i
    \qquad\text{and}\qquad
    p(S)=\sum_{i\in S}p_i,
\]
respectively. For each color $c\in\colors$, we further define
\[
    \items_c=\{i\in\items:\kappa_i=c\}
    \qquad\text{and}\qquad
    k_c(S)=|S\cap\items_c|,
\]
that is, the set of items of color $c$ and the number of such items in $S$.

The \emph{Colored Knapsack Problem} (\ColKP) asks for a maximum-profit subset of items $S\subseteq\items$ with $w(S)\le b$ that can be ordered so that no two consecutive items share the same color. As shown by~\citet{Borges2024}, such an ordering exists if and only if, for every color, its number of selected items is at most one greater than the number of selected items of all other colors, \ie{} if and only if
\begin{equation}
    k_c(S)\;\le\;|S\setminus \items_c|+1,\qquad\, c\in\colors.
    \label{eq:colorcons}
\end{equation}
For a subset $S\subseteq\items$, let
\[
    d(S)=\max\{k_c(S):c\in\colors\}.
\]
A color $c\in\colors$ is \emph{dominant} in $S$ if $k_c(S)=d(S)$, that is, if it has the largest number of selected items. Condition~\eqref{eq:colorcons} is then equivalent to the single inequality
\begin{equation}
    2\,d(S)\;\le\;|S|+1 ,
    \label{eq:feas}
\end{equation}
which states that the number of items of a dominant color cannot exceed the number of all selected items by more than one. Accordingly, the \ColKP{} can be stated compactly as
\begin{equation}
    \OPT \;=\; \max\bigl\{\,p(S)\;:\; S\subseteq \items,\; w(S)\le b,\;
    2 \, d(S)\le |S|+1 \,\bigr\}.
    \label{eq:ColKP}
\end{equation}
We call a subset $S$ satisfying~\eqref{eq:feas} \emph{color-feasible}, and \emph{feasible} if in addition $w(S)\le b$. A dominant color $c$ is \emph{critical} for $S$ if $k_c(S)=|S\setminus \items_c|+1$. Observe that at most one color can be critical for a given feasible \ColKP{} solution $S$.

Two features distinguish the \ColKP{} from the classical Knapsack Problem (\KP) and are essential to our analysis. First, its family of feasible sets is not closed under taking subsets: removing an item may destroy color feasibility. Consequently, items with zero or even negative profit may belong to an optimal solution, because they can act as \emph{separators} that allow additional profitable items of a dominant color to be selected. We therefore formulate the \ColKP{} with integer profits $p_i\in\mathbb{Z}$ for all $i\in\items$, rather than restricting profits to be positive.

Example~\ref{ex:running} illustrates a \ColKP{} instance with $n=5$ items, $m=2$ colors, and capacity $b=9$. If colors are ignored, $\{1,2,3\}$ is an optimal \KP{} solution with profit $19$. Nonetheless, it contains three items of color $1$ and none of color $2$, and is therefore infeasible for the \ColKP{}, since $2d=6>|S|+1=4$. $S^\star=\{2,4,5\}$ is instead an optimal \ColKP{} solution with profit $\OPT=17$. Color $2$ is its unique dominant color, with $d(S^\star)=2$, and it is also critical because $2d(S^\star)=|S^\star|+1$.

\begingroup
\renewcommand{\figurename}{Example}
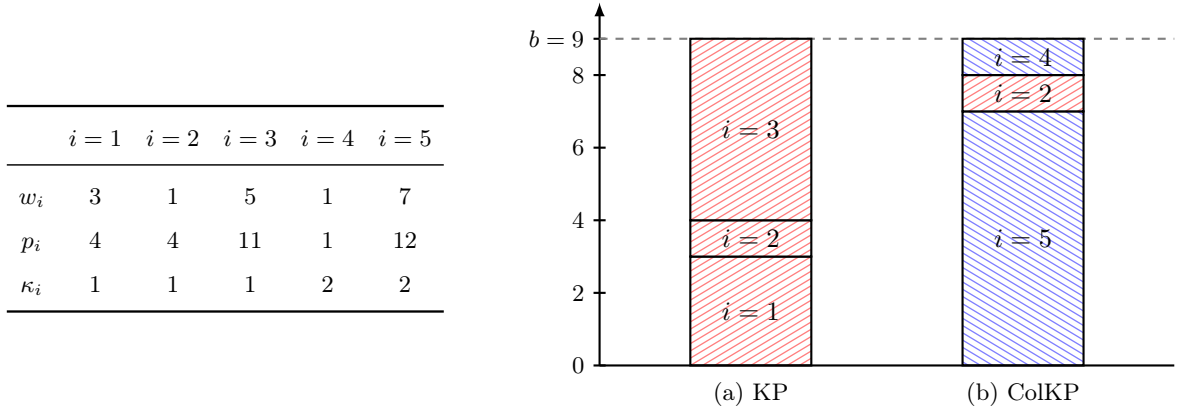
\begin{figure}[]
\centering
\begin{minipage}[c]{0.42\textwidth}
\centering
\footnotesize
\renewcommand \arraystretch{1.5}
\setlength{\tabcolsep}{4.5pt}
\begin{tabular}{cccccc}
\toprule
& $i = 1$ & $i = 2$ & $i = 3$ & $i = 4$ & $i = 5$\\
\midrule
$w_i$    & 3 & 1 & 5 & 1 & 7\\
$p_i$    & 4 & 4 & 11 & 1 & 12\\
$\kappa_i$ & 1 & 1 & 1 & 2 & 2\\
\bottomrule
\end{tabular}
\end{minipage}%
\hfill
\begin{minipage}[c]{0.58\textwidth}
\centering
\begin{tikzpicture}[y=0.6cm, x=1cm, scale=0.8]
    \draw[->, thick, >=latex] (-1.5,0) -- (-1.5,10);
    \foreach \y in {0,2,4,6,8} {
        \draw[thick] (-1.6,\y) -- (-1.4,\y)
            node[left, xshift=-4pt] {\footnotesize $\y$};
    }
    \draw[thick] (-1.6,9) -- (-1.4,9)
        node[left, xshift=-4pt] {\footnotesize $b=9$};
    \draw[dashed, gray, thick] (-1.5,9) -- (8,9);

    \node at (1,-0.8) {\footnotesize (a) KP};
    \begin{scope}
        \clip (0,0) rectangle (2,3);
        \fill[white] (0,0) rectangle (2,3);
        \foreach \i in {-3,-2.8,...,3} {
            \draw[red!50, line width=0.5pt] (\i,0) -- (\i+3,3);
        }
    \end{scope}
    \draw[thick] (0,0) rectangle (2,3);
    \node at (1,1.5) {\small $i=1$};

    \begin{scope}
        \clip (0,3) rectangle (2,4);
        \fill[white] (0,3) rectangle (2,4);
        \foreach \i in {-1,-0.8,...,4} {
            \draw[red!50, line width=0.5pt] (\i,3) -- (\i+1,4);
        }
    \end{scope}
    \draw[thick] (0,3) rectangle (2,4);
    \node at (1,3.5) {\small $i=2$};

    \begin{scope}
        \clip (0,4) rectangle (2,9);
        \fill[white] (0,4) rectangle (2,9);
        \foreach \i in {-5,-4.8,...,3} {
            \draw[red!50, line width=0.5pt] (\i,4) -- (\i+5,9);
        }
    \end{scope}
    \draw[thick] (0,4) rectangle (2,9);
    \node at (1,6.5) {\small $i=3$};

    \node at (5.5,-0.8) {\footnotesize (b) ColKP};
    \begin{scope}
        \clip (4.5,0) rectangle (6.5,7);
        \fill[white] (4.5,0) rectangle (6.5,7);
        \foreach \i in {0,0.2,...,15} {
            \draw[blue!50, line width=0.5pt] (\i,0) -- (\i-7,7);
        }
    \end{scope}
    \draw[thick] (4.5,0) rectangle (6.5,7);
    \node at (5.5,3.5) {\small $i=5$};

    \begin{scope}
        \clip (4.5,7) rectangle (6.5,8);
        \fill[white] (4.5,7) rectangle (6.5,8);
        \foreach \i in {3,3.2,...,8} {
            \draw[red!50, line width=0.5pt] (\i,7) -- (\i+1,8);
        }
    \end{scope}
    \draw[thick] (4.5,7) rectangle (6.5,8);
    \node at (5.5,7.5) {\small $i=2$};

    \begin{scope}
        \clip (4.5,8) rectangle (6.5,9);
        \fill[white] (4.5,8) rectangle (6.5,9);
        \foreach \i in {4,4.2,...,9} {
            \draw[blue!50, line width=0.5pt] (\i,8) -- (\i-1,9);
        }
    \end{scope}
    \draw[thick] (4.5,8) rectangle (6.5,9);
    \node at (5.5,8.5) {\small $i=4$};

\draw[thick] (-1.5,0) -- (8,0);
\end{tikzpicture}
\end{minipage}
\caption{\ColKP{} instance with $n=5$ items, $m=2$ colors and capacity $b=9$. The item weights, profits, and colors are given in the table on the left. The figure on the right illustrates an optimal \KP{} solution and an optimal \ColKP{} solution to the instance. The vertical axis represents consumed capacity, and the height of each rectangle is the corresponding item weight. Red and blue hatching denote colors $1$ and $2$, respectively. Ignoring colors, the KP optimum in~(a) selects items $1$, $2$, and $3$. The \ColKP{} optimum in~(b) selects items $5$, $2$, and $4$, thus respecting the color constraints.}
\label{ex:running}
\end{figure}
\endgroup

\medskip

 A two-color version of the \ColKP{} was studied as a subproblem of the Black and White Bin Packing Problem (see~\citet{Balogh2013, Balogh2015}). Nonetheless, the results obtained for the two-color case do not extend to the $m$-color case. The general variant was instead introduced by~\citet{CHM2026}. To the best of our knowledge, their paper is the only previous work devoted specifically to this problem. They proposed a natural integer linear programming (ILP) formulation and studied its linear programming (LP) relaxation, proving that an optimal basic LP solution has at most two fractional items and can be computed in $O(n)$ time. They also established that the \ColKP{} is weakly $\mathcal{NP}$-hard and developed two exact pseudo-polynomial dynamic programming (DP) algorithms. The first processes the items one at a time and runs in $O(b\,n^4)$ time, whereas the second processes one color class at a time and runs in $O(b^2\,n^3)$ time. They derived dominance and fathoming rules for both algorithms and compared their implementations computationally with off-the-shelf solvers applied to the natural \ColKP{} ILP formulation.

That study settled the basic complexity status of the \ColKP{} and provided pseudo-polynomial exact algorithms, but left its approximability open: indeed, neither a polynomial-time approximation algorithm with a provable performance guarantee nor an approximation scheme was proposed. The \ColKP{} inherits weak $\mathcal{NP}$-hardness from the classical \KP{}, since a \KP{} instance can be viewed as a \ColKP{} instance where each item has a distinct color. However, the two pseudo-polynomial DPs proposed in~\cite{CHM2026} do not immediately yield an approximation scheme, nor is it clear which of the two recursions is the more suitable starting point to develop such an algorithm. This paper addresses both issues: it gives the first approximation algorithm for the \ColKP{}, to the best of our knowledge, and derives a fully polynomial-time approximation scheme from each of the two DPs of~\cite{CHM2026}.

\subsection{Approximation algorithms and approximation schemes for knapsack problems}
\label{sec:approximation-background}

In this section, we recall the definitions and results on approximation algorithms and approximation schemes that will be used throughout the paper. We also review the classical scaling argument underlying the fully polynomial-time approximation schemes for the \KP{} and several of its variants, which we later adapt to the \ColKP{} in Sections~\ref{sec:profitdp} and~\ref{sec:fptas}.

We first recall the performance measures used throughout the paper. Let
$\mathcal A$ be an algorithm for a maximization problem. For an instance $I$, let $\mathcal A(I)$ be the value of the feasible solution it returns and let $\OPT(I) > 0$ be the instance optimum. We assume that $0\le\mathcal A(I)\le\OPT(I)$. The \emph{absolute worst-case performance ratio} of $\mathcal A$ is
\begin{equation}
    R_{\mathcal A}
    :=\inf_{I}
      \left\{\frac{\mathcal A(I)}{\OPT(I)} \right\}
    \in[0,1].
    \label{eq:absolute-wcpr}
\end{equation}
If $\mathcal A$ runs in polynomial time and $R_{\mathcal A}\ge\alpha$ for some $\alpha\in(0,1]$, then $\mathcal A$ is an \emph{$\alpha$-approximation algorithm}. A guarantee $\alpha$ is \emph{tight} when $R_{\mathcal A}=\alpha$. 

\medskip

A \emph{polynomial-time approximation scheme} (PTAS) is a family of algorithms $\{\mathcal A_\varepsilon\}_{\varepsilon\in(0,1)}$ such that
\begin{equation}
    \mathcal A_\varepsilon(I)\ge(1-\varepsilon) \, \OPT(I)
    \label{eq:ptas-definition}
\end{equation}
for every instance $I$, and such that, for every fixed $\varepsilon$, the running time of $\mathcal A_\varepsilon$ is polynomial in the input size. A \emph{fully polynomial-time approximation scheme} (FPTAS) is a PTAS whose running time is polynomial both in the input size and in $1/\varepsilon$.

\medskip

The classical \KP{} is the canonical example of a weakly $\mathcal{NP}$-hard problem admitting an FPTAS (see~\citet{KPP2004} for a detailed overview). The standard scaling-based scheme starts from the pseudo-polynomial dynamic program indexed by profit. Let
\[
    W_i(z):=\min\bigl\{w(S):S\subseteq\{1,2,\dots,i\},\ p(S)=z\bigr\}
\]
be the minimum weight of a subset of the first $i$ items with profit exactly $z$. In case no such subset exists, we set $W_i(z)=+\infty$. The optimum is the largest profit $z$ for which $W_n(z)\le b$. 
Let $\pmax=\max\{p_i : i\in\items\}$ and set
\begin{equation}
    \scalefactor=\frac{\varepsilon\pmax}{n}
    \qquad \text{and}\qquad
    \hat p_i=\left\lfloor\frac{p_i}{\scalefactor}\right\rfloor,~~i \in \items.
    \label{eq:classical-scaling}
\end{equation}
For $S\subseteq\items$, write $\hat p(S)=\sum_{i\in S}\hat p_i$.
Running the same recursion on the scaled profits makes its profit range polynomially bounded, since $\sum_{i \in \items}\hat p_i\le n^2/\varepsilon$. If $S_\varepsilon$ is optimal for the scaled instance and $S^\star$ is optimal for the original one, then
\begin{equation}
\begin{aligned}
    p(S_\varepsilon)
    &\ge \scalefactor\,\hat p(S_\varepsilon)
     \ge \scalefactor\,\hat p(S^\star) \\
    &\ge p(S^\star)-n\scalefactor
     \ge (1-\varepsilon) \, \OPT,
\end{aligned}
    \label{eq:classical-scaling-bound}
\end{equation}
where the last inequality uses $\pmax\le\OPT$. Scaling therefore trades an additive loss of less than $\scalefactor$ per selected item for a profit-indexed table of polynomial size. This is the basic idea that we adapt to the \ColKP{} in Sections~\ref{sec:profitdp} and~\ref{sec:fptas}.

\citet{Sahni1975} gave the first PTAS for the \KP{}, based on the enumeration of a bounded number of high-profit items. The first published FPTAS is instead due to \citet{IbarraKim1975}. \citet{Lawler1979} refined that scheme and improved its complexity bounds, and \citet{MagazineOguz1981} proposed a further dynamic-programming-based scheme.
Later work continued improving the running time and, in some cases, the space requirements~\cite{KellererPferschy1999,KellererPferschy2004,Chan2018}. A recent sequence of improvements~\cite{Jin2019,DengJinMao2023} led to the $\widetilde O(n+1/\varepsilon^2)$-time FPTAS of~\citet{Chen2025}. 

FPTASs are also known for several variants of the \KP{}; for a broader survey, see~\citet{Cacchiani2022}. Binary and integer versions of the \emph{Multiple-Choice Knapsack Problem} were studied from an approximation perspective by~\citet{ChandraHirschbergWong1976}. In their binary version, the items are partitioned into classes and at most one item is selected from each class. The integer version permits instead multiple copies of the selected item type. An early exact study of the exact-one binary variant, including two branch-and-bound algorithms, is due to~\citet{Nauss1978}, while~\citet{Lawler1979} proposed an FPTAS for the multiple-choice problem. For the $k$-item \KP{}, which imposes an upper bound on the cardinality of the selected set, \citet{Caprara2000} gave a pseudo-polynomial dynamic program, a PTAS, and an FPTAS. More involved side constraints require problem-specific arguments. The \emph{Knapsack Problem with Conflicts} is strongly $\mathcal{NP}$-hard on general conflict graphs, whereas \citet{PferschySchauer2009} derived pseudo-polynomial dynamic programs, and hence FPTASs, for bounded-treewidth and chordal conflict graphs. A more recent exact branch-and-bound algorithm is given by \citet{Coniglio2021}. The \emph{Knapsack Problem with Setup} was studied by~\citet{ChajakisGuignard1994}, who proposed exact algorithms, including a pseudo-polynomial DP. \citet{ChebilKhemakhem2015} developed a further DP, while~\citet{FuriniMonaciTraversi2018} embedded tailored LP-relaxation algorithms in an implicit-enumeration scheme and ran it in parallel with an improved DP. Moreover, \citet{PferschyScatamacchia2018} proved an inapproximability result for the general problem and derived FPTASs for three relevant special cases. Finally, \citet{DAmbrosio2018} introduced the \emph{Product Knapsack Problem}, whose objective is to maximize the product of the selected item profits, which may be positive or negative, and proposed mixed-integer linear and nonlinear formulations and a pseudo-polynomial exact DP algorithm. \citet{PferschySchauerThielen2021} subsequently gave the first FPTAS for this problem.

\subsection{Contributions and outline of the paper}

Section~\ref{sec:structure} first establishes the structural properties needed to handle arbitrary integer profits. It shows that, for any \ColKP{} instance, items with profit at most $-\pmax$ can be removed in linear time without changing the optimal value, so that all relevant profits belong to the $(-\pmax,\pmax]$ interval. It also proves that an optimal solution containing a negative-profit item must have a unique critical color and can contain at most $(n-1)/2$ negative-profit items. These properties make it possible to bound the signed profit ranges used later.

Section~\ref{sec:constant} then develops the first approximation algorithm for the \ColKP{}. Starting from an optimal basic solution of the natural \ColKP{} formulation LP relaxation, we round down its fractional variables (at most two) and, when necessary, repair color feasibility by removing one item. Comparing the resulting set with the most profitable singleton yields a $1/3$--approximation in $O(n)$ time. We also show that the absolute worst-case performance ratio of this algorithm is exactly $1/3$. Besides being a result of independent interest, the algorithm provides lower and upper bounds on $\OPT$ that differ by at most a factor of three and are therefore crucial for controlling the scaled profit range.

Sections~\ref{sec:profitdp} and~\ref{sec:fptas} build the main approximation scheme. We first reformulate the item-by-item dynamic program $\firstDP$ of~\cite{CHM2026} by exchanging the roles of capacity and profit. We then combine this profit-indexed recursion with scaling and the aforementioned bounds to obtain an FPTAS that runs in $O(n^5/\varepsilon)$ time for nonnegative profits and in $O(n^6/\varepsilon)$ time for arbitrary integer profits.

Section~\ref{sec:secondDP} investigates whether the alternative color-by-color recursion $\secondDP$ of~\cite{CHM2026} can yield a faster approximationscheme. Its profit-indexed reformulation also leads to an FPTAS, but the resulting running times are $O(n^5/\varepsilon^2)$ for nonnegative profits and $O(n^7/\varepsilon^2)$ for arbitrary integer profits, showing that the item-by-item recursion provides a better basis for the direct scaling approach developed here.

Finally, Section~\ref{sec:conclusions} summarizes the results and outlines directions for future research.

\section{Structural properties of \ColKP{} solutions}
\label{sec:structure}

Throughout this section, and in the rest of the paper, we let
\begin{equation}
    \pmax \;=\; \max\left\{p_i:i\in\items\right\}
    \label{eq:pmax}
\end{equation}
denote the largest item profit. Moreover, in the remainder of the paper we assume without loss of generality that $w_i\le b$ for all $i\in \items$. 

\begin{observation}\label{obs:singleton}
The empty set $S = \emptyset$ is a feasible solution to the \ColKP, and so is every singleton $S=\{i\}$ with $i\in \items$,
because $|S|=d(S)=1$ satisfies~\eqref{eq:feas} and $w_i\le b$ for all $i\in \items$.
Consequently
\begin{equation}
    \OPT\;\ge\;\pmax\;\ge\;0 .
    \label{eq:OPTgepmax}
\end{equation}
Moreover, if $\pmax\le0$ then the empty set is optimal and problem~\eqref{eq:ColKP} is solved in $O(n)$ time. We therefore assume $\pmax>0$ from now on.
\end{observation}

Observation~\ref{obs:singleton} is the reason why an FPTAS is meaningful at all for a problem with possibly negative profits: the optimal value is guaranteed to be nonnegative, and in fact bounded from below by the largest item profit, which is exactly the quantity that drives the classical scaling argument for the classical \KP.

The next lemma describes two ways in which a feasible \ColKP{} solution can be modified without losing color feasibility. It is used repeatedly in the rest of the paper.

\begin{lemma}
\label{lem:removal}
Let $S$ be a feasible \ColKP{} solution, let $t=|S|$ and $d=d(S)$.
\begin{enumerate}[label=(\alph*)]
    \item If $i\in S$ belongs to a dominant color of $S$, then $S\setminus\{i\}$ is color-feasible.
    \item If $i\in S$ belongs to a dominant color of $S$ and $j\in S$ belongs to a color that is \emph{not} dominant in $S$, then $S\setminus\{i,j\}$ is
    color-feasible.
\end{enumerate}
\end{lemma}

\begin{proof}

(a) Let $c=\kappa_i$, so that $k_c(S)=d$, and let $S'=S\setminus\{i\}$, with $|S'|=t-1$. If some color $c'\neq c$ also satisfies $k_{c'}(S)=d$, then
$t\ge 2d$ and $d(S')=d$, so that $2d(S')=2d\le t=|S'|+1$. Otherwise $c$ is the unique dominant color of $S$, every other color appears at most $d-1$ times, and $d(S')=d-1$. From feasibility of $S$, we have that $2d\le t+1$, hence $2d(S')=2d-2\le t-1 = |S'|$.

\medskip

\noindent (b) Let $S'=S\setminus\{i,j\}$, with $|S'|=t-2$, and let $c=\kappa_i$. If some color $c''\neq c$ satisfies $k_{c''}(S)=d$, then $S$ contains $d$ items
of color $c$, $d$ items of color $c''$ and the item $j$, whose color is not dominant by assumption, hence $t\ge2d+1$. Since $d(S')\le d$ we obtain $2d(S')\le 2d\le t-1=|S'|+1$. Otherwise $c$ is the unique dominant color of $S$, all colors appear at most $d-1$ times in $S'$, and $2d(S')\le2d-2\le t-1 =|S'|+1$, where the middle inequality follows from $2d\le t+1$.
\end{proof}

\noindent Part~(b) of Lemma~\ref{lem:removal} immediately bounds how negative a \emph{relevant} profit can be.

\begin{lemma}
\label{lem:verynegative} 
There exists an optimal solution of~\eqref{eq:ColKP} containing no item $i \in \items$ with $p_i\le -\pmax$. 
\end{lemma}

\begin{proof}
Let $S$ be an optimal solution minimizing the number of items with $p_i\le-\pmax$ and suppose by contradiction that $i\in S$ satisfies $p_i\le -\pmax$. Note that $S\neq\{i\}$, since $p(\{i\})\le-\pmax<0\le\OPT$.
If $\kappa_i \in \colors$ is a dominant color of $S$, then $S\setminus\{i\}$ is feasible by Part~(a) of Lemma~\ref{lem:removal} and $p(S\setminus\{i\})=p(S)-p_i\ge p(S)+\pmax>p(S)$, contradicting optimality. Otherwise, pick any item $j\in S$ of a dominant color of $S$; then $j\neq i$ and, by Part~(b) of Lemma~\ref{lem:removal}, $S'=S\setminus\{i,j\}$ is feasible, with $p(S')=p(S)-p_i-p_j\ge p(S)+\pmax-\pmax=p(S)$. Hence $S'$ is again optimal and contains
fewer items with profit at most $-\pmax$ than $S$, a contradiction.
\end{proof}

An immediate consequence of Lemma~\ref{lem:verynegative}, which will be often used in the following, is that all items $i \in \items$ such that $p_i\le -\pmax$ can be deleted in $O(n)$ time, and after this preprocessing every item satisfies
\begin{equation}
    -\pmax\;<\;p_i\;\le\;\pmax .
    \label{eq:profitrange}
\end{equation}
For this reason, we will assume~\eqref{eq:profitrange} to hold in the remainder of the paper.

\medskip

We now present the last result of this section, which explains when negative items can be used at all in an optimal \ColKP{} solution.

\begin{lemma}
\label{lem:tight}
Let $S$ be an optimal solution containing an item $i$ with $p_i<0$. Then $S$ has a unique dominant color, which is also critical, \ie{} $|S|=2d(S)-1$.
Furthermore, $S$ contains at most $d(S)-1\le(n-1)/2$ items of negative profit.
\end{lemma}

\begin{proof}
Let $t=|S|$, $d=d(S)$, and $\sigma=t+1-2d\ge0$. Suppose by contradiction that $\sigma\ge1$.
If $\kappa_i \in \colors$ is a dominant color, then $S\setminus\{i\}$ is color-feasible by Part~(a) of Lemma~\ref{lem:removal}. If it is not, then $d(S\setminus\{i\})=d$ and $2d\le t+1-\sigma\le t=|S\setminus\{i\}|+1$, so $S\setminus\{i\}$ is color-feasible as well. In both cases $w(S\setminus\{i\})\le w(S)\le b$ and $p(S\setminus\{i\})=p(S)-p_i>p(S)$, contradicting the optimality of $S$. Hence $\sigma=0$, \ie{} $t=2d-1$, meaning that the dominant color is critical and unique. Moreover, no items of the critical color can have negative profit, thus $S$ contains at most $t-d=d-1\le(n-1)/2$ items of negative profit.
\end{proof}

\section{A linear-time \texorpdfstring{$1/3$}{1/3}-approximation}
\label{sec:constant}

The scaling parameter of many FPTASs of knapsack type is $\scalefactor=\varepsilon\,\LB/n$, where $\LB$ is a lower bound on $\OPT$. After scaling the item profits, the size of the dynamic programming table is proportional to $\UB/\scalefactor$, where $\UB$ is an upper bound on $\OPT$. The running time of the scheme therefore depends on the ratio $\UB/\LB$. Taking $\LB=\pmax$ and $\UB=z_{LP}$ (the value of the LP relaxation) may be a natural choice, but it only guarantees $\UB/\LB=O(n)$. In this section we show that a constant ratio can be achieved in linear time, exploiting the results introduced in~\cite{CHM2026}.

We recall the LP relaxation of the \ColKP{},
\begin{equation}
\begin{aligned}
  (\mathrm{LP_{Col}})\qquad z_{LP}=\max\Bigl\{\, \sum_{i \in \items} p_i \, x_i : &\sum_{i \in \items} w_i \, x_i \le b,\\ &\sum_{i\in \items_c}x_i-\sum_{i\in \items\setminus \items_c}x_i\le1,~~ c\in\colors,\\&x_i \in [0,1],~~ i \in \items\Bigr\},
\label{eq:lp}  
\end{aligned}
\end{equation}
together with the two properties established in~\cite{CHM2026}: every basic feasible solution to~$\mathrm{LP_{Col}}$ has at most two fractional variables (because at most one color constraint can be active at a feasible point), and an optimal solution can be computed in $O(n)$ time (by first solving the \KP{} relaxation to identify a color that is critical in some optimal solution to~$\mathrm{LP_{Col}}$, and then applying the multidimensional search algorithm by \citet{Megiddo1984} to a resulting two-constraint linear program). Selecting an extreme-point optimum in the final reduced linear program preserves this running time, so the optimal solution used below may be assumed basic.

In Algorithm~\ref{alg:round}, we describe an LP-based primal algorithm, referred to as the \textsc{LP-Round} algorithm, which rounds down an optimal basic LP solution and repairs a possible violation of the color constraints by removing one item from the solution.

\begin{algorithm}[]
\caption{The \textsc{LP-Round} algorithm for the \ColKP{}}
\label{alg:round}
\begin{algorithmic}[1]
\Require a \ColKP{} instance
\Ensure a feasible \ColKP{} solution $\hat S\subseteq\items$
\State $\pmax\leftarrow\max \{p_i:i\in\items\}$
\State compute an optimal basic solution $\{x_1^\ast, x_2^\ast, \dots, x_n^\ast \}$ to $\mathrm{LP_{Col}}$ 
\State $S_0\leftarrow\{i\in\items:x^\ast_i=1\}$
\If{$S_0$ violates~\eqref{eq:feas}}
    \State let $c$ be the unique dominant color and let
           $j\in\argmin\{p_i:i\in S_0\cap\items_c\}$
    \State $S_0\leftarrow S_0\setminus\{j\}$ \Comment{repair}
\EndIf
\State let $r\in\argmax \{p_i:i\in\items\}$
\State $\hat S\in\argmax\{p(S):S\in\{S_0,\{r\}\}\}$
\State \textbf{return} $\hat S$
\end{algorithmic}
\end{algorithm}

\begin{theorem}\label{thm:third}
Algorithm~\ref{alg:round} returns a feasible \ColKP{} solution $\hat{S}$ with
\[
    p(\hat{S})\;\ge\;\tfrac13\,z_{LP}\;\ge\;\tfrac13\,\OPT ,
\]
and runs in $O(n)$ time.
\end{theorem}

\begin{proof}
Let $\{x_1^\ast, x_2^\ast, \dots, x_n^\ast \} \in \mathbb{Q}^n$, be an optimal basic solution to~$\mathrm{LP_{Col}}$, let
$F=\{i\in\items:0<x_i^\ast<1\}$, and put $S_0=\{i\in\items:x_i^\ast=1\}$. By~\cite{CHM2026}, $|F|\le2$, and therefore $\sum_{i\in F}x_i^\ast < 2$. Moreover,
$w(S_0)\le\sum_{i\in\items}w_i \, x_i^\ast\le b$, and
\begin{equation}
    p(S_0)
    =z_{LP}-\sum_{i\in F}p_i \, x_i^\ast
    \ge z_{LP}-2\pmax .
    \label{eq:roundloss}
\end{equation}
For every $c\in\colors$, we have that
\[
k_c(S_0) = \sum_{i \in \items_c} x_i^\ast - \sum_{i \in F \cap \items_c} x_i^\ast \qquad \text{and}\qquad |S_0 \setminus \items_c| = \sum_{i \in \items \setminus \items_c} x_i^\ast - \sum_{i \in F \setminus \items_c} x_i^\ast \, .
\]
It follows that
\begin{align*}
    k_c(S_0)-|S_0\setminus\items_c|
    &=\underbrace{\left(\sum_{i\in\items_c}x_i^\ast
      -\sum_{i\in\items\setminus\items_c}x_i^\ast\right)}_{\le 1 \text{ by feasibility of the LP solution}}
      -\sum_{i\in F\cap\items_c}x_i^\ast
      +\sum_{i\in F\setminus\items_c}x_i^\ast \\[1.5ex]
      & \le  \, 1 + \sum_{i \in F \setminus \items_c} x_i^\ast \, < \, 3 \, .
\end{align*}
Since the left-hand side is an integer, we have that $k_c(S_0)-|S_0\setminus\items_c|\le2$ for every color $c\in\colors$, thus $S_0$ violates each color constraint by at most one unit.

If $S_0$ is color-feasible, no repair is needed and~\eqref{eq:roundloss} applies directly. Suppose otherwise, and let $c\in\colors$ be the violated color. Such a color is unique: if two distinct colors were violated, each would occur at least $(|S_0|+2)/2$ times, which is impossible. Letting $d=k_c(S_0)$, we therefore have
\[
    |S_0\setminus\items_c|=d-2
    \quad\implies\quad
    |S_0|=2d-2 \, .
\]
Removing any item $j\in S_0\cap\items_c$ restores color feasibility. Indeed, the resulting set has $2d-3$ items, color $c$ occurs $d-1$ times, and every other color occurs at most $d-2$ times. Hence the color constraint associated with color $c$ is tight, since $2(d-1)=(2d-3)+1$, and all other color constraints are satisfied as well.

We now bound the amount of profit lost when this repair is needed. We first determine the structure of the fractional variables. The contribution of the integral variables to the color constraint of $c$ is
\[
    k_c(S_0)-|S_0\setminus\items_c|=d-(d-2)=2.
\]
Separating the integral and fractional variables in that constraint, and by feasibility of the LP solution, we obtain
\[
    2+\sum_{i\in F\cap\items_c}x_i^\ast
      -\sum_{i\in F\setminus\items_c}x_i^\ast\le1,
\]
or, equivalently,
\begin{equation}
    \sum_{i\in F\setminus\items_c}x_i^\ast
      -\sum_{i\in F\cap\items_c}x_i^\ast\ge1.
    \label{eq:fractionalbalance}
\end{equation}
Recall that $|F|\le2$ and that $0<x_i^\ast<1$ for every $i\in F$. If $|F|\le1$, the left-hand side of~\eqref{eq:fractionalbalance} is strictly
smaller than one. The same is true if $|F|=2$ and at least one of the two items belongs to $\items_c$: in that case at most one positive term appears
in the first sum, and that term is strictly smaller than one. Therefore
\[
    F=\{a,b\},\quad a,b\notin\items_c,
    \qquad\text{and}\qquad x_a^\ast+x_b^\ast\ge1.
\]

We now show that the color constraint of $c$ is active for the LP solution. To this end, consider any color $h\neq c$. Since all items of color $h$ in $S_0$ belong to $S_0\setminus\items_c$, we have $k_h(S_0)\le d-2$. Moreover, $|S_0|=2d-2$. We write the left-hand side of the color constraint of $h$ at the LP solution by separating the integral and fractional variables:
\begin{align*}
    &\sum_{i\in\items_h}x_i^\ast
      -\sum_{i\in\items\setminus\items_h}x_i^\ast \\
    &\quad=\left(k_h(S_0)+\sum_{i\in F\cap\items_h}x_i^\ast\right)
      -\left(|S_0\setminus\items_h|
      +\sum_{i\in F\setminus\items_h}x_i^\ast\right) \\
    &\quad=k_h(S_0)-\biggl(|S_0|-k_h(S_0)\biggr)
      +\sum_{i\in F\cap\items_h}x_i^\ast
      -\sum_{i\in F\setminus\items_h}x_i^\ast \\
    &\quad=2 \, k_h(S_0)-|S_0|
      +\sum_{i\in F\cap\items_h}x_i^\ast
      -\sum_{i\in F\setminus\items_h}x_i^\ast \\
    &\quad\le 2(d-2)-(2d-2)
      +\sum_{i\in F\cap\items_h}x_i^\ast
      -\sum_{i\in F\setminus\items_h}x_i^\ast \\
    &\quad\le-2+\sum_{i\in F\cap\items_h}x_i^\ast<0 \, .
\end{align*}
Hence, every color constraint other than that of $c$ is slack.

We further observe that the $n-2$ integral variables $x_i^\ast$, $i \in \items \setminus F$, are at one of their bounds and yield $n-2$ linearly independent active bound constraints. Since the LP solution is basic, there must exist two additional linearly independent active constraints. These are the capacity constraint and the color constraint associated with $c$, which can be the only tight color constraint. Since $a,b\notin\items_c$, its left-hand side is $2-x_a^\ast-x_b^\ast$, and hence
\[
    2-x_a^\ast-x_b^\ast=1,
    \quad \implies \quad x_a^\ast+x_b^\ast=1.
\]
We can now bound the profit loss. We have that
\[
    p_a \, x_a^\ast+p_b \, x_b^\ast
      \le\pmax \, (x_a^\ast+x_b^\ast)=\pmax.
\]
Therefore
\[
    p(S_0)=z_{LP}-p_a \, x_a^\ast-p_b \,x_b^\ast
      \ge z_{LP}-\pmax.
\]
The repair removes an item $j\in S_0\cap\items_c$. Since $p_j\le\pmax$, the
repaired set satisfies
\[
    p(S_0\setminus\{j\})=p(S_0)-p_j
      \ge z_{LP}-2\pmax.
\]
If no repair is needed, the same lower bound follows directly from~\eqref{eq:roundloss}. Thus, in either case, the set obtained by rounding the LP solution, and possibly repairing the solution obtained, has profit at least $z_{LP}-2 \, \pmax$.

Finally, we notice that Algorithm~\ref{alg:round} compares the solution obtained with a singleton of profit $\pmax$. If $\pmax\ge z_{LP}/3$, the singleton has profit at least $z_{LP}/3$. If $\pmax<z_{LP}/3$, the rounded and possibly repaired set has profit
\[
    z_{LP}-2 \, \pmax
      >z_{LP}-\tfrac23 \, z_{LP}=\tfrac13 \, z_{LP}.
\]
In both cases $p(\hat S)\ge z_{LP}/3$. Since~$\mathrm{LP_{Col}}$ is a relaxation, $z_{LP}\ge\OPT$, and the claimed approximation guarantee follows.

All operations other than solving~$\mathrm{LP_{Col}}$ take $O(n)$ time. Moreover, an optimal basic solution to~$\mathrm{LP_{Col}}$ can also be computed in $O(n)$ time (as shown in~\cite{CHM2026}), which proves the running time claim.
\end{proof}

\begin{corollary}\label{cor:roundtight}
The worst-case performance ratio of \textsc{LP-Round} is exactly $1/3$.
\end{corollary}

\begin{proof}
Theorem~\ref{thm:third} shows that the ratio is at least $1/3$. To prove that this bound is tight, let $\ell\ge2$ be an integer and consider the following
instance with two colors and four items:
\[
{
\begin{array}{cccccc}
\toprule
& i=1 & i=2 & i=3 & i=4 \\
\midrule
w_i & 1 & \ell & \ell-1 & 2\ell \\
p_i & \ell & \ell & \ell-1 & \ell \\
\kappa_i & 1 & 1 & 1 & 2 \\
\bottomrule
\end{array}
}
\qquad\text{and}\qquad b=3\ell.
\]
Set $t=1-1/(3\ell)$. We first show that
\[
    x_1^\ast=1,\qquad x_2^\ast=t,\qquad x_3^\ast=0,\qquad x_4^\ast=t
\]
is an optimal solution to $\mathrm{LP}_{\mathrm{Col}}$ for this instance. This solution is feasible: we have $0<t<1$, its weight is $1+3\ell t=3\ell$, and the color constraint of color~$1$ is tight, since $1+t-t=1$. The color constraint of color~$2$ is respected as well, since $t-(1+t)=-1$.

We prove optimality by deriving two upper bounds on the profit of an arbitrary feasible choice of values $x_i$, $i\in\items$. First, comparing its profit with its weight gives
\begin{align}
    \sum_{i\in\items}p_i \, x_i
    &=\sum_{i\in\items}w_i \, x_i+(\ell-1)x_1-\ell x_4 \notag\\
    &\le 3\ell+(\ell-1)-\ell x_4
     =4\ell-1-\ell x_4 .
    \label{eq:tight-cap-bound}
\end{align}
The equality follows by substituting the weights and profits of the four items. The inequality follows from the capacity constraint $\sum_{i\in\items}w_i \, x_i\le3\ell$ and from $x_1\le1$.

A second upper bound is obtained by rewriting the profit of this solution in terms of the left-hand side of the color constraint associated with color~$1$, thus obtaining
\begin{align}
    \sum_{i\in\items}p_i \, x_i
    &=\ell(x_1+x_2+x_3-x_4)-x_3+2\ell x_4 \notag\\
    &\le \ell+2\ell x_4 .
    \label{eq:tight-color-bound}
\end{align}
Here the inequality follows from the color constraint $x_1+x_2+x_3-x_4\le1$ and from $x_3\ge0$.

The bound in~\eqref{eq:tight-cap-bound} decreases with $x_4$, whereas the bound in~\eqref{eq:tight-color-bound} increases with $x_4$. They coincide at
$x_4=t$: indeed,
\[
    4\ell-1-\ell t=\ell+2\ell t=3\ell-\frac23.
\]
Therefore, if $x_4\le t$,~\eqref{eq:tight-color-bound} gives $\sum_{i\in\items}p_i \, x_i\le3\ell-2/3$; if $x_4\ge t$, the same conclusion follows from~\eqref{eq:tight-cap-bound}. Since $\sum_{i\in\items}p_i \, x_i^\ast=\ell(1+2t)=3\ell-2/3$, the proposed solution is optimal.
We now show that the proposed solution is the unique optimal solution as well. As observed above, a solution of profit $3\ell-2/3$ must have $x_4=t$. Equality in~\eqref{eq:tight-cap-bound} then requires $\sum_{i\in\items}w_i \, x_i=3\ell$ and $x_1=1$, while equality in~\eqref{eq:tight-color-bound} requires $x_3=0$ and $x_1+x_2+x_3-x_4=1$. It follows that $x_2=x_4=t$. Hence $x_i=x_i^\ast$ for all $i\in\items$, proving uniqueness. We also observe that the optimal solution is basic.

Rounding down the optimal solution gives $S_0=\{1\}$, which is already color-feasible, so no repair is performed. Since $\pmax=\ell$, \textsc{LP-Round} returns a solution $\hat S_\ell$ of profit $\ell$.

On the other hand, $S^\ast_\ell=\{1,3,4\}$ is color-feasible, has weight $3\ell$, and has profit $3\ell-1$. This solution is optimal. Indeed, every set with at most two items has profit at most $2 \, \ell$, while a color-feasible solution with three items must contain item~$4$ and two items of color~$1$. The other two such solutions apart from $S^\ast_\ell$, $\{1,2,4\}$ and $\{2,3,4\}$, have weight $3\ell+1$ and $4\ell-1$, respectively, and are therefore infeasible. Finally, the set of all four items violates both the capacity constraint and the color constraint of color~$1$. Therefore
\[
    \frac{p(\hat S_\ell)}{p(S^\ast_\ell)}
      =\frac{\ell}{3\ell-1}
      \xrightarrow[\ell\to\infty]{}\frac13.
\]
Together with Theorem~\ref{thm:third}, this proves the claim.
\end{proof}

We now run Algorithm~\ref{alg:round} on the instance of Example~\ref{ex:running}. An optimal solution to the \KP{} relaxation assigns $x_1=0$, $x_2=1$, $x_3=1$, $x_4=0$, and $x_5=\tfrac37$, and has value $141/7$. The color constraint of color~$1$ reads $2-\tfrac37=\tfrac{11}{7}>1$ and is violated, so by Lemma~2 of~\cite{CHM2026} color~$1$ is critical in some optimal solution to~$\mathrm{LP_{Col}}$. Solving the resulting linear program, obtained by removing all the color constraints except the one for color~$1$, gives
\[
    x_1^\ast=0,\quad x_2^\ast=1,\quad x_3^\ast=1,\quad x_4^\ast=\tfrac23,\quad x_5^\ast=\tfrac13,
    \qquad \text{and} \qquad
    z_{LP}=\frac{59}{3} \, .
\]
The two fractional variables are associated with items $4$ and $5$, both of color $ \kappa_4 = \kappa_5 = 2$, which illustrates that two fractional variables of an optimal basic solution may share a (non-critical) color.
Rounding down gives $S_0=\{2,3\}$, two items of color~$1$ and none of color $2$: the color constraint is violated by exactly one unit, since $k_1(S_0)=2=|S_0\setminus\items_1|+2$. The repair removes the least profitable item of color~$1$ in $S_0$, namely item~$2$, leaving $\{3\}$ of profit $11$, which is feasible. The best singleton is item~$5$, of profit $\pmax=12$, and the algorithm returns it, thus $\LB=12$. This is a lower bound on the optimal value $\OPT=17$ of the instance, which is attained by the feasible solution $\{2,4,5\}$, and it also shows that the ratio $\LB/\OPT$ can be, in practice, much better than the worst-case guarantee of $1/3$.

\medskip

\noindent From now on we denote by
\begin{equation}
    \LB\;=\;p(\hat{S})\qquad\text{and}\qquad \UB\;=\;3 \cdot \LB
    \label{eq:bounds}
\end{equation}
the lower and upper bounds provided by Theorem~\ref{thm:third}, so that
$\pmax\le\LB\le\OPT\le\UB$.

\section{A first profit-indexed DP algorithm}
\label{sec:profitdp}

In this section we describe the exact DP algorithm on which the first approximation scheme proposed in this paper is built.
It coincides with the item-by-item dynamic program $\firstDP$ of~\cite{CHM2026} with the roles of the capacity and of the objective exchanged.

Assume the items are sorted by color, so that each color class occupies a contiguous block of indices. The state of $\firstDP$ after processing the first $i$ items records the number $t$ of selected items, the number $d$ of selected items of a dominant color, the number $a$ of selected items of the current color $\kappa_i \in \colors$, and the consumed capacity $q$. Its value is the maximum profit. The key observation is that the color feasibility of a partial solution $S\subseteq\{1, 2, \dots,i\}$, and of every extension of $S$ by items of $\{i+1, i+2, \dots,n\}$, is completely determined by the triple $(t,d,a)=(|S|, \, d(S), \, k_{\kappa_i}(S))$ and by the items still to be processed. It does not depend on $w(S)$ nor on $p(S)$. Hence, for two partial solutions with the same $(i,t,d,a)$ and the same profit, the one of smaller weight can never be worse. This dominance property makes the classical exchange of dimensions legitimate. We observe that this is not automatic for KP variants with combinatorial side constraints: it would fail, for instance, under a classwise lower bound on the total selected weight, as in one variant of the KP with group fairness~\cite{Malaguti2025}, since then a lighter partial solution would not dominate a heavier one.

\medskip

\noindent For $i\in\{1,2,\dots,n\}$, $t,d,a\in\{0,1,\dots,n\}$, and an integer profit value $z$, define
\begin{equation}
    g(i,t,d,a,z)=\min\Bigl\{w(S)\;:\;S\subseteq\{1,2,\dots,i\},\;|S|=t,\;
    d(S)=d,\;k_{\kappa_i}(S)=a,\;p(S)=z\Bigr\},
    \label{eq:gdef}
\end{equation}
with $g=+\infty$ when no such subset exists. Let $\sigma(a,d)=1$ if $a=d$ and $\sigma(a,d)=0$ otherwise. At stage $i=0$, we use the initialization convention $g(0,0,0,0,0)=0$, with all other states having value equal to $+\infty$. Algorithm~\ref{alg:firstDP-fptas} computes $g$ by the two usual transitions (pack / do not pack item $i$), resetting the current-color counter at every
color change.

\begin{algorithm}[]
\caption{\textsc{$\firstDPp$}: profit-indexed dynamic program}
\label{alg:firstDP-fptas}
\begin{algorithmic}[1]
\Require a \ColKP{} instance, with items indexed by nondecreasing color index
\Ensure an optimal \ColKP{} solution $S^\star\subseteq\items$ and its value
        $p(S^\star)=\OPT$
\State $p_{\mathrm{abs}}\leftarrow\sum_{i\in\items}|p_i|$; \quad
       $\mathcal Z\leftarrow\{-p_{\mathrm{abs}},-p_{\mathrm{abs}}+1, \dots,p_{\mathrm{abs}}\}$
\State $g(i,t,d,a,z)\leftarrow+\infty$ for all
       $(i,t,d,a,z)\in\{0,1,\dots,n\}^4\times\mathcal Z$
\State  $g(0,0,0,0,0)\leftarrow0$
\For{$i=1,\dots,n$}
  \For{all $(t,d,a,z)$ with $g(i-1,t,d,a,z)<+\infty$}
     \State $a'\leftarrow a$ \textbf{if} $i>1$ and $\kappa_i=\kappa_{i-1}$,
            \textbf{else} $a'\leftarrow0$
     \State $g(i,t,d,a',z)\leftarrow\min\{g(i,t,d,a',z),\,g(i-1,t,d,a,z)\}$
            \Comment{do not pack $i$}
     \State $d'\leftarrow d+\sigma(a',d)$
     \State $g(i,t+1,d',a'+1,z+p_i)\leftarrow
            \min\{g(i,t+1,d',a'+1,z+p_i),\,g(i-1,t,d,a,z)+w_i\}$
            \Comment{pack $i$}
  \EndFor
\EndFor
\State $z^\star\leftarrow\max\{z\in\mathcal Z\;:\;\exists\,(t,d,a)\ \text{ with }
       2d\le t+1\ \text{and}\ g(n,t,d,a,z)\le b\}$
\State choose $(t^\star,d^\star,a^\star)$ with $2d^\star\le t^\star+1$ and
       $g(n,t^\star,d^\star,a^\star,z^\star)\le b$
\State recover $S^\star$ by backtracking the transitions leading to $g(n,t^\star,d^\star,a^\star,z^\star)$
\State \textbf{return} $S^\star$ and $z^\star$
\end{algorithmic}
\end{algorithm}

\begin{proposition}\label{prop:dpcorrect}
Algorithm~\ref{alg:firstDP-fptas} computes $g(i,t,d,a,z)$ as defined in~\eqref{eq:gdef} for every state, and returns an optimal \ColKP{} solution.
\end{proposition}

\begin{proof}
We proceed by induction on $i$. For $i=0$ the only subset is $\emptyset$, with $t=d=a=z=0$ and weight $0$, which is the initialization.
Assume the claim holds for $i-1$. Since the items are sorted by color, either $i=1$ or $\kappa_i=\kappa_{i-1}$ (the current color continues, and the counter
$a$ is inherited) or $\kappa_i\neq\kappa_{i-1}$ (item $i$ is the first of a new color, and the counter is reset to $a'=0$). Let $S\subseteq\{1,2,\dots,i\}$ realize a state $(i,t,d,a'',z)$. If $i\notin S$, then $S\subseteq\{1,2,\dots,i-1\}$ realizes the state $(i-1,t,d,a,z)$ for the appropriate value of $a$ and is propagated by the first transition, with unchanged $t$, $d$ and $z$. If $i\in S$, then $S\setminus\{i\}$ realizes a state at stage $i-1$ with $t-1$ items, profit $z-p_i$, weight $w(S)-w_i$, current-color counter $a'=a''-1$ and dominant cardinality $d(S\setminus\{i\})$; since $d(S)=\max\{d(S\setminus\{i\}),a'+1\}$ and $a'\le d(S\setminus\{i\})$, we have $d(S)=d(S\setminus\{i\})+\sigma(a',d(S\setminus\{i\}))$, which is exactly the update of line~6. Moreover $S\setminus\{i\}$ must be of minimum weight among the subsets realizing its state, otherwise replacing it would contradict the minimality of $w(S)$. By the induction hypothesis its weight is stored at stage $i-1$, and the second transition adds $w_i$. Since the two cases are exhaustive and mutually exclusive and the algorithm takes minima over all incoming transitions, the claim follows. The output is correct because a set $S$ is feasible if and only if $2d(S)\le|S|+1$ and $w(S)\le b$, and for each profit level it is enough to know the minimum weight realizing it.
\end{proof}

It is worth noting that the state-reduction rules described in~\cite{CHM2026} can be adapted to $\firstDPp$, but they are optional and are not used in the correctness or running time results below.

\section{A first approximation scheme}
\label{sec:fptas}

We are now ready to state and analyze the first approximation scheme we propose in this article, summarized in Algorithm~\ref{alg:fptas}.
The idea is the same as in the classical FPTAS for the \KP: profits are scaled so that the profit dimension of the DP table becomes polynomial in $n$ and $1/\varepsilon$. More precisely, given the lower bound $\LB$ of~\eqref{eq:bounds}, we set $\scalefactor=\varepsilon\LB/n$ and replace every profit $p_i$ by $\hat p_i=\lfloor p_i/\scalefactor\rfloor$, $i \in \items$. It can be shown that the total loss caused by rounding down the scaled profits is then at most $n\scalefactor=\varepsilon\LB\le\varepsilon\OPT$. At the same time, the number of scaled profit levels in the DP tabledepends on the value $\UB/\scalefactor$, and hence by the ratio $\UB/\LB$. 

\begin{algorithm}[]
\caption{FPTAS for the \ColKP}
\label{alg:fptas}
\begin{algorithmic}[1]
\Require {a \ColKP{} instance and an
accuracy $\varepsilon\in(0,1)$}
\Ensure {a feasible \ColKP{} solution $\hat S\subseteq\items$}
\State let $\pmax=\max \{p_i:i\in\items\}$
       \Comment{Lemma~\ref{lem:verynegative}}
\State $S^H\leftarrow\textsc{LP-Round}$; $\LB\leftarrow p(S^H)$;
       $\UB\leftarrow3\LB$ \Comment{Theorem~\ref{thm:third}}
\State $\scalefactor\leftarrow\varepsilon\LB/n$; \quad
       $\hat p_i\leftarrow\lfloor p_i/\scalefactor\rfloor$ for all $i\in \items$
\If{$p_i\ge0$ for all $i\in\items$}
    \State $\zmin\leftarrow0$; \quad $\zmax\leftarrow\lceil\UB/\scalefactor\rceil$
\Else
    \State $\zmin\leftarrow\lfloor-n\pmax/\scalefactor\rfloor$; \quad
           $\zmax\leftarrow\lceil(\UB+n\pmax)/\scalefactor\rceil$
\EndIf
\State sort the items by color
\State run Algorithm~\ref{alg:firstDP-fptas} with profits $\hat p_i$, $i \in \items$, restricted to profit levels $z\in[\zmin,\zmax]$ and to states with $g\le b$
\State let $\hat S$ be a feasible set of maximum $\hat p(\hat S)$ found, and
       recover it by backtracking
\State \textbf{return} $\hat S$
\end{algorithmic}
\end{algorithm}

The following proposition establishes the formal correctness and approximation guarantee of the algorithm.

\begin{proposition}\label{thm:guarantee}
For every $\varepsilon\in(0,1)$, Algorithm~\ref{alg:fptas} returns a feasible \ColKP{} solution of value at least $(1-\varepsilon)\OPT$.
\end{proposition}

\begin{proof}
We observe that \eqref{eq:bounds} gives $\LB\ge\pmax>0$, so the scaling factor $\scalefactor=\varepsilon\LB/n$ is well defined and positive. Moreover, every item has profit at most $\pmax$, and hence $\OPT\le n\pmax$. Since $\LB\le\OPT$ and $\varepsilon<1$, it follows that
\begin{equation}
    \scalefactor=\frac{\varepsilon\LB}{n}
      \le\varepsilon\pmax<\pmax,
    \qquad\text{and thus}\qquad \frac{\pmax}{\scalefactor}>1.
    \label{eq:scalefactorltpmax}
\end{equation}

Let $S^\star$ be an optimal solution to a \ColKP{} instance. For every item $i$, the definition $\hat p_i=\lfloor p_i/\scalefactor\rfloor$ yields
\[
    \hat p_i\le \frac{p_i}{\scalefactor}<\hat p_i+1,
\]
or, equivalently, $\scalefactor\hat p_i\le p_i<\scalefactor(\hat p_i+1)$. Summing these inequalities over an arbitrary set $S\subseteq\items$ gives
\begin{equation}
    \scalefactor\hat p(S)\,\le\,p(S)\,\le\,\scalefactor\hat p(S)+|S|\scalefactor\,\le\,\scalefactor\hat p(S)+n\scalefactor \, .
    \label{eq:roundset}
\end{equation}

We now verify that truncating the profit axis does not remove any state needed to represent $S^\star$. Recall that the items have been sorted by color before the dynamic program is run. For every $i\in\{1,2,\dots,n\}$, let $S_i^\star=S^\star\cap\{1,2,\dots,i\}$ be the prefix of $S^\star$ containing the first $i$ items, and let $N=\{j\in S^\star:p_j<0\}$ be the set of items in $S^\star$ with negative profit. Every $j\in N$ satisfies $|p_j|<\pmax$ by~\eqref{eq:profitrange}. Hence
$
    \sum_{j\in N}|p_j|<n\pmax
$.
Furthermore, the profit of every prefix $S_i^\star$ is at least the sum of all the negative profits in $S^\star$ and at most the sum of all the positive profits in $S^\star$. Therefore
\begin{equation}
    -n\pmax < \sum_{j \in N} p_j
    \le p(S_i^\star)
    \le\sum_{j\in S^\star\setminus N}p_j
    =\sum_{j\in S^\star}p_j + \sum_{j\in N}|p_j|
    <\UB+n\pmax \, .
    \label{eq:prefix-profit-bounds}
\end{equation}

We now want to prove that the scaled profit of every partial solution $S_i^\star$ belongs to the interval $[\zmin,\zmax]$. Suppose first that some item profit is negative. If $N=\emptyset$, the scaled profit of each prefix is nonnegative. Otherwise, Lemma~\ref{lem:tight} gives
$|N|\le(n-1)/2$. For every $j\in N$,
\[
    \hat p_j=\left\lfloor\frac{p_j}{\scalefactor}\right\rfloor
      >\frac{p_j}{\scalefactor}-1
      >-\frac{\pmax}{\scalefactor}-1,
\]
while $\hat p_j\ge0$ for every nonnegative-profit item. Consequently, for every partial solution $S_i^\star$,
\begin{align*}
    \hat p(S_i^\star)
      &>-\frac{n-1}{2}\left(\frac{\pmax}{\scalefactor}+1\right) \\
      &>-\frac{n\pmax}{\scalefactor}
       \ge\left\lfloor-\frac{n\pmax}{\scalefactor}\right\rfloor=\zmin.
\end{align*}
On the other hand,~\eqref{eq:roundset} and
\eqref{eq:prefix-profit-bounds} give
\[
    \hat p(S_i^\star)
      \le\frac{p(S_i^\star)}{\scalefactor}
      <\frac{\UB+n\pmax}{\scalefactor}
      \le\left\lceil\frac{\UB+n\pmax}{\scalefactor}\right\rceil=\zmax.
\]
Thus the scaled profit of every partial solution belongs to $[\zmin,\zmax]$. Observe that, if all profits are nonnegative, $\hat p_i \ge 0$ for all $i \in \items$, and the tighter bounds used by the algorithm follow directly from
\[
    0\le\hat p(S_i^\star)
      \le\hat p(S^\star)
      \le\frac{p(S^\star)}{\scalefactor}
      =\frac{\OPT}{\scalefactor}\le\frac{\UB}{\scalefactor} \, .
\]
We therefore conclude that the scaled profit of every partial solution belongs to $[\zmin,\zmax]$ in this case as well.

We now want to show that the proposed dynamic program returns a solution whose scaled profit is at least $\hat p(S^\star)$. For every $i\in\{1,2,\dots,n\}$, since all weights are positive, we have that $w(S_i^\star)\le w(S^\star)\le b$: consequently, the state corresponding to $S_i^\star$ is not discarded for exceeding the capacity. For $i=n$, the resulting state has profit $\hat p(S_n^\star)=\hat p(S^\star)$ and satisfies $2d(S^\star)\le|S^\star|+1$, because $S^\star$ is feasible by assumption. Hence $\hat p(S^\star)$ is one of the terminal profit levels over which the dynamic program maximizes. Since $\hat S$ is chosen with maximum scaled profit among those terminal states, we obtain
\begin{equation}
    \hat p(\hat S)\ge\hat p(S^\star) \, .
    \label{eq:scaled-dominance}
\end{equation}

The value returned by Algorithm~\ref{alg:fptas} is therefore the value of a feasible solution and, using~\eqref{eq:roundset} and~\eqref{eq:scaled-dominance}, we obtain
\begin{align*}
    p(\hat S)
    &\ge \scalefactor\hat p(\hat S) \\
    &\ge \scalefactor\hat p(S^\star) \\
    &> p(S^\star)-|S^\star|\scalefactor \\
    &\ge \OPT-n\scalefactor \\
    &= \OPT-\varepsilon\LB \\
    &\ge (1-\varepsilon) \,\OPT.
\end{align*}
\end{proof}

After proving the approximation guarantee, we now bound the size of the scaled dynamic-programming table and derive the running time of the proposed approximation scheme.

\begin{theorem}\label{thm:time}
Algorithm~\ref{alg:fptas} runs in
$O (n^5 /\varepsilon)$ time when all profits are nonnegative, and in $O(n^6/\varepsilon)$ time for arbitrary integer profits. 
\end{theorem}

\begin{proof}
Consider first the case $p_i\ge0$ for all $i\in\items$, in which $\zmin=0$ and the profit axis has
\[
    \Pi\;=\;\zmax +1\;=\;O\!\left(\frac{\UB}{\scalefactor}\right)
    \;=\;O\!\left(\frac{3\LB\,n}{\varepsilon\LB}\right)
    \;=\;O\!\left(\frac{n}{\varepsilon}\right)
\]
levels.

In Algorithm~\ref{alg:firstDP-fptas} the state components $t$ and $d$ take $O(n)$ values each, while the counter $a$ of the current color is bounded by $|\items_c|$ for all colors $c \in \colors$. Summing over the items grouped by color, the number of states is
\[
    \sum_{c\in\colors}\sum_{i\in \items_c}O\bigl(n\cdot n\cdot|\items_c|\cdot\Pi\bigr)
    \;=\;O\!\left(n^2\Pi\sum_{c\in\colors}|\items_c|^2\right)
    \;=\;O\!\left(\frac{n^3}{\varepsilon}\sum_{c\in\colors}|\items_c|^2\right),
\]
and each state is processed with $O(1)$ transitions. Since
$\sum_c|\items_c|^2\le\bigl(\sum_c|\items_c|\bigr)^2=n^2$, the bound $O(n^5/\varepsilon)$ follows.

\medskip

For arbitrary integer profits, the profit axis has instead
\[
    \Pi=\zmax-\zmin+1
    \le \frac{\UB + 2n\pmax}{\scalefactor}+3
    = \frac{3 \LB n + 2n^2\pmax}{\varepsilon \LB}+3
    =O\left(\frac{n^2}{\varepsilon}\right).
\]
The running time therefore becomes
$O(n^6/\varepsilon)$.
\end{proof}

The additional factor $n$ comes from the need to retain every prefix of an optimal solution. With nonnegative profits, the profit of a partial solution increases monotonically and remains between $0$ and the profit of the complete solution, so using the interval $[0,\UB/\scalefactor]$ ensures that all possible profits are covered. With arbitrary profits, a partial solution may instead have negative profit or may exceed the final profit before later negative-profit items are processed. Both ends of the scaled profit axis must therefore include an additional contribution of order $n\pmax/\scalefactor$. 

\medskip

We conclude the section by illustrating the behaviour of the FPTAS detailed in Algorithm~\ref{alg:fptas} on the instance depicted in Example~\ref{ex:running}. We first set $\varepsilon=1/2$.
Algorithm~\ref{alg:round} returns $\LB=12$ on this instance, as described in Section~\ref{sec:constant}, so $\scalefactor=\varepsilon\LB/n=6/5$, the scaled profits are those of Table~\ref{tab:example}, and the profit axis is truncated at $\zmax=\lceil\UB/\scalefactor\rceil=30$. The largest scaled profit reached by a feasible solution is $z=13$, attained both by $\{2,5\}$, of weight $8$, and by the optimal solution $S^\star=\{2,4,5\}$, of weight $9$. These sets correspond to distinct terminal DP states, so either may be returned depending on tie-breaking. In particular, $\{2,5\}$ is a possible output, of true profit $16$. This illustrates the effect of scaling: $\hat p_4=0$, so item~4 can become invisible in the scaled objective although it still consumes capacity. The performance guarantee holds, since $16\ge(1-\varepsilon)\OPT=\frac{17}{2}$.

Consider instead $\varepsilon=1/100$. The scaling factor becomes $\scalefactor=\varepsilon\LB/n=3/125$, the scaled profits are given in the last row of Table~\ref{tab:example}, and the profit axis is truncated at $\zmax=\lceil\UB/\scalefactor\rceil=1500$. Now item~4 has scaled profit $41$, rather than zero. The largest scaled profit reached by a feasible solution is $707$, uniquely attained by the optimal solution $S^\star=\{2,4,5\}$. Such profit is strictly larger than the scaled profit $666$ of $\{2,5\}$. Hence the DP returns $S^\star$. In fact, the optimality of the returned solution also follows directly from the approximation analysis: it gives $p(\hat S)>\OPT-\varepsilon\LB=\frac{422}{25}$, and, since all solution profits are integral, this implies $p(\hat S)\ge17=\OPT$. Thus, for this instance, reducing $\varepsilon$ not only tightens the guaranteed ratio, but makes the FPTAS exact, at the cost of a profit axis that is fifty times longer.

\begin{table}[H]
\centering
\caption{Scaling of the item profits in the instance of Example~\ref{ex:running}, with $\LB=12$, for $\varepsilon=1/2$ and $\varepsilon=\frac{1}{100}$.}
\medskip
\label{tab:example}
\begin{tabular}{cccccc}
\toprule
 & $i=1$ & $i=2$ & $i=3$ & $i=4$ & $i=5$\\
\midrule
$p_i$ & 4 & 4 & 11 & 1 & 12\\
$\hat p_i$ ($\varepsilon=1/2$) & 3 & 3 & 9 & 0 & 10\\
$\hat p_i$ ($\varepsilon=\frac{1}{100}$) & 166 & 166 & 458 & 41 & 500\\
\bottomrule
\end{tabular}
\end{table}

\section{Can the color-by-color recursion do better?}
\label{sec:secondDP}

Neither of the exact dynamic programs proposed in~\cite{CHM2026} dominates the other: $\firstDP$ runs in $O(b \, n^4)$ and $\secondDP$ in $O(b^2 \, n^3)$. Thus $\secondDP$ has the smaller asymptotic bound when $b=o(n)$ and is no worse in order when $b=O(n)$. It is therefore natural to ask whether $\secondDP$ may provide a basis for a faster approximation scheme than the one based on $\firstDP$. We show that the direct profit-scaling implementation considered here does not improve the worst-case running time, and explain the source of the additional factor.

Recall the structure of $\secondDP$. For each color $c\in\colors$, an \emph{inner} dynamic program enumerates the relevant packings of the items of that color. An \emph{outer} dynamic program then combines one packing per color, in the spirit of a Multiple-Choice Knapsack Problem, keeping track of the number of packed items and of the number of items of a dominant color. $\secondDP$ is indexed by the amount of consumed capacity, and its value corresponds to a solution profit. We denote by $\secondDPp$ the DP obtained by exchanging the two dimensions, exactly as $\firstDPp$ was obtained from $\firstDP$ in Section~\ref{sec:profitdp}: the inner and the outer DPs are therefore indexed by the profit, and their value is the minimum weight of a solution achieving that profit. The exchange can be performed here for the same reason as in the case of $\firstDPp$. By~\eqref{eq:feas}, the color feasibility of a set $S\subseteq\items$ depends only on the color cardinalities $k_c(S)$, $c\in\colors$. Hence, replacing a packing $S_c\subseteq\items_c$ chosen for a color $c \in \colors$ by any other packing $S_c'\subseteq\items_c$ with $|S_c'|=|S_c|$ and $p(S_c')=p(S_c)$ leaves the cardinalities and the profit of the overall solution unchanged, and does not increase its weight if $w(S_c')\le w(S_c)$. Within a color class, therefore, we need to keep only one minimum-weight packing per pair cardinality-profit.

Specifically, for every color $c\in\colors$, cardinality $k\in\{0,1,\dots,|\items_c|\}$ and integer profit level $\zeta$, let
\begin{equation}
    h_c(k,\zeta)=\min\bigl\{w(S)\;:\;S\subseteq\items_c,\;
      |S|=k,\;p(S)=\zeta\bigr\},
    \label{eq:innerdp}
\end{equation}
where $h_c(k,\zeta)=+\infty$ if no solution with $k$ items and profit $\zeta$ exists. This table can be computed by the usual item-by-item knapsack DP with an additional cardinality counter, with the roles of capacity and profit exchanged.
Following the notation of~\cite{CHM2026}, we then define
\[
    \mathcal S_c=\bigl\{(k,\zeta)\in
       \{0,1,\dots,|\items_c|\}\times \mathbb Z\;:\;
       h_c(k,\zeta)<+\infty\bigr\}
\]
that is, a set containing all combinations of number of items $k$ and profit $\zeta$ for which there exists a feasible packing of the items having color $c \in \colors$. The weight associated with a pair $(k,\zeta)$ is $h_c(k,\zeta)$. In particular, $(0,0)\in\mathcal S_c$ represents an empty solution.

For $c\in\{1,2,\dots,m\}$, $t,d\in\{0,1,\dots,n\}$ and integer profit level $z$, the outer program of $\secondDPp$ then computes
\begin{equation}
\begin{split}
    g(c,t,d,z)=\min\Bigl\{\,\sum_{j=1}^c h_j(k_j,\zeta_j)\;:\; 
       &(k_j,\zeta_j)\in\mathcal S_j~\text{ for all }~j \in \{1,2,\dots,c\}, \\
       &\sum_{j=1}^c k_j=t,\quad
       \max_{1\le j\le c}k_j=d,\quad
       \sum_{j=1}^c\zeta_j=z\Bigr\}.
\end{split}
    \label{eq:outerdp}
\end{equation}
We initialize the empty solution as $g(0,0,0,0)=0$, with all other states having value equal to $+\infty$. For each state $(c-1,t,d,z)$ such that $g(c-1,t,d,z)<+\infty$, and each $(k,\zeta)\in\mathcal S_c$, the forward transition is
\begin{equation}
\begin{split}
    g(c,t+k,\max\{d,k\},z+\zeta)\leftarrow\min\bigl\{&
       g(c,t+k,\max\{d,k\},z+\zeta),\\
       &g(c-1,t,d,z)+h_c(k,\zeta)\bigr\}.
\end{split}
    \label{eq:outertransition}
\end{equation}
After running the outer DP, the maximum profit level with $g(m,t,d,z)\le b$ and $2d\le t+1$ for some $t,d \in \{0,1,\dots,n\}$ is therefore the optimal profit, and a corresponding solution can be recovered by backtracking through the outer and inner programs.

Algorithm~\ref{alg:secondDP-profit} reports the complete pseudocode of the profit-indexed color-by-color dynamic program $\secondDPp$.

\begin{algorithm}[]
\caption{\textsc{$\secondDPp$}: profit-indexed color-by-color dynamic program}
\label{alg:secondDP-profit}
\begin{algorithmic}[1]
\Require a \ColKP{} instance
\Ensure an optimal \ColKP{} solution $S^\star\subseteq\items$ and its value
        $p(S^\star)=\OPT$
\State $p_{\mathrm{abs}}\leftarrow\sum_{i\in\items}|p_i|$; \quad
       $\mathcal Z\leftarrow\{-p_{\mathrm{abs}},-p_{\mathrm{abs}}+1,\dots,p_{\mathrm{abs}}\}$
\For{$c=1,\dots,m$}
  \State compute $\mathcal S_c$ and the values $h_c(k,\zeta)$ for $(k,\zeta)\in\mathcal S_c$
\EndFor
\State $g(c,t,d,z)\leftarrow+\infty$ for all
       $(c,t,d,z)\in\{0,\dots,m\}\times\{0,\dots,n\}^2\times\mathcal Z$
\State $g(0,0,0,0)\leftarrow0$
\For{$c=1,\dots,m$}
  \For{all $(t,d,z)$ with $g(c-1,t,d,z)<+\infty$}
    \For{all $(k,\zeta)\in\mathcal S_c$}
      \State $t'\leftarrow t+k$; \quad $d'\leftarrow\max\{d,k\}$;
             \quad $z'\leftarrow z+\zeta$
      \State $\bar w\leftarrow g(c-1,t,d,z)+h_c(k,\zeta)$
      \State $g(c,t',d',z')\leftarrow\min\{g(c,t',d',z'),\bar w\}$
    \EndFor
  \EndFor
\EndFor
\State $z^\star\leftarrow\max\{z\in\mathcal Z:\exists\,(t,d)\text{ with }
       2d\le t+1\text{ and }g(m,t,d,z)\le b\}$
\State choose $(t^\star,d^\star)$ with $2d^\star\le t^\star+1$ and
       $g(m,t^\star,d^\star,z^\star)\le b$
\State recover $S^\star$ by backtracking through the outer and inner transitions
\State \textbf{return} $S^\star$ and $z^\star$
\end{algorithmic}
\end{algorithm}

The next proposition states that $\secondDPp$ is exact. 
\begin{proposition}\label{prop:secondDPcorrect}
Algorithm~\ref{alg:secondDP-profit} computes $g(c,t,d,z)$ as defined in~\eqref{eq:outerdp} for every state, and returns an optimal \ColKP{} solution.
\end{proposition}

\begin{proof}
For every color $c\in\colors$, a preliminary item-by-item DP computes $h_c(k,\zeta)$ as defined in~\eqref{eq:innerdp}; consequently, $\mathcal S_c$ contains exactly the attainable cardinality--profit pairs for color $c$.

We now prove the claim on the outer DP by induction on $c$. For $c=0$ the only admissible choice in~\eqref{eq:outerdp} is the empty one, which gives $g(0,0,0,0)=0$ and $g(0,t,d,z)=+\infty$ for every other state. Assume that the claim holds for the states of stage $c-1$, and consider a state $(c,t,d,z)$.

We first show that the value computed at $(c,t,d,z)$ is at most the right-hand side of~\eqref{eq:outerdp}. Let $(k_j,\zeta_j)_{j=1}^{c}$ be admissible for $(c,t,d,z)$ and let
\[
    t'=t-k_c,\qquad d'=\max_{1\le j\le c-1}k_j,\qquad
    z'=z-\zeta_c .
\]
The prefix $(k_j,\zeta_j)_{j=1}^{c-1}$ is admissible for the state $(c-1,t',d',z')$, so that $g(c-1,t',d',z')\le\sum_{j=1}^{c-1}h_j(k_j,\zeta_j)<+\infty$ by the induction hypothesis. Moreover $t'+k_c=t$, $\max\{d',k_c\}=d$ and $z'+\zeta_c=z$. Hence the transition~\eqref{eq:outertransition} applied to the state $(c-1,t',d',z')$ with the pair $(k_c,\zeta_c)\in\mathcal S_c$ writes into $(c,t,d,z)$, and
\[
    g(c,t,d,z)\;\le\;g(c-1,t',d',z')+h_c(k_c,\zeta_c)
      \;\le\;\sum_{j=1}^{c}h_j(k_j,\zeta_j).
\]
Taking the minimum over all admissible choices gives the desired inequality.

Conversely, every finite value written into $(c,t,d,z)$ is of the form $g(c-1,t',d',z')+h_c(k,\zeta)$, for some state $(c-1,t',d',z')$ of finite value and some pair $(k,\zeta)\in\mathcal S_c$ with $t'+k=t$, $\max\{d',k\}=d$ and $z'+\zeta=z$. By the induction hypothesis, $g(c-1,t',d',z')$ is attained by a choice $(k_j,\zeta_j)_{j=1}^{c-1}$ admissible for $(c-1,t',d',z')$. Setting $(k_c,\zeta_c)=(k,\zeta)$ extends it to a choice that is admissible for $(c,t,d,z)$ and whose value is exactly the value written. No state of the table has therefore value  smaller than the right-hand side of~\eqref{eq:outerdp}, and the two coincide. 

It remains to examine the output. Let $(t,d,z)$ be such that $g(m,t,d,z)<+\infty$, let $(k_c,\zeta_c)_{c=1}^{m}$ be an admissible choice attaining that value, and let $S_c\subseteq\items_c$ be a minimum-weight solution featuring $k_c$ items, with a profit of $\zeta_c$, and a total weight of $h_c(k_c,\zeta_c)$. The set $S=\bigcup_{c\in\colors}S_c$ satisfies $k_c(S)=k_c$ for every color $c\in\colors$, and hence
\[
    |S|= \sum_{c=1}^m k_c,\qquad d(S)=\max_{c\in\colors}k_c=d,\qquad p(S)= \sum_{c = 1}^m \zeta_c=z,\qquad
    w(S)=g(m,t,d,z).
\]
A profit level $z$ is therefore obtained by a feasible solution if and only if $g(m,t,d,z)\le b$ and $2d\le t+1$ for some $t,d\in\{0,1,\dots,n\}$, so that the largest such $z$ equals $\OPT$ and backtracking returns an optimal solution.
\end{proof}

For the approximation scheme, both the inner and the outer program of $\secondDPp$ are run on the scaled item profits $\hat p_i$, $i\in\items$, and both are restricted to the profit levels $\mathcal Z=\{\zmin,\zmin+1,\dots,\zmax\}$, with endpoints as in Algorithm~\ref{alg:fptas}. Let $\Pi=|\mathcal Z| = \zmax - \zmin + 1$. States with a weight greater than $b$ are discarded, so that a transition of both the inner and the outer DP is performed only if the profit level it produces belongs to $\mathcal Z$ and the weight it produces is at most $b$. In the next proposition, we show that these restrictions preserve the states needed for the approximation guarantee.
The resulting scheme is detailed in Algorithm~\ref{alg:secondDP-fptas}.

\begin{algorithm}[]
\caption{FPTAS for the \ColKP{} based on $\secondDPp$}
\label{alg:secondDP-fptas}
\begin{algorithmic}[1]
\Require {a \ColKP{} instance and an accuracy $\varepsilon\in(0,1)$}
\Ensure {a feasible \ColKP{} solution $\hat S\subseteq\items$}
\State let $\pmax=\max \{p_i : i\in\items\}$
       \Comment{Lemma~\ref{lem:verynegative}}
\State $S^H\leftarrow\textsc{LP-Round}$; $\LB\leftarrow p(S^H)$;
       $\UB\leftarrow3\LB$ \Comment{Theorem~\ref{thm:third}}
\State $\scalefactor\leftarrow\varepsilon\LB/n$; \quad
       $\hat p_i\leftarrow\lfloor p_i/\scalefactor\rfloor$ for all $i\in\items$
\If{$p_i\ge0$ for all $i\in\items$}
    \State $\zmin\leftarrow0$; \quad $\zmax\leftarrow\lceil\UB/\scalefactor\rceil$
\Else
    \State $\zmin\leftarrow\lfloor-n\pmax/\scalefactor\rfloor$; \quad
           $\zmax\leftarrow\lceil(\UB+n\pmax)/\scalefactor\rceil$
\EndIf
\State run Algorithm~\ref{alg:secondDP-profit} with profits $\hat p_i$,
       restricting both inner and outer tables to profit levels
       $z\in[\zmin,\zmax]$ and to states of weight at most $b$
\State let $\hat S$ be a feasible set of maximum $\hat p(\hat S)$ found, and
       recover it by backtracking
\State \textbf{return} $\hat S$
\end{algorithmic}
\end{algorithm}

\begin{proposition}\label{prop:secondDP}
For every $\varepsilon\in(0,1)$, Algorithm~\ref{alg:secondDP-fptas} returns a
feasible \ColKP{} solution of value at least $(1-\varepsilon)\OPT$.
\end{proposition}

\begin{proof}
By positivity of $\pmax$, we have $\LB\ge\pmax>0$, while \eqref{eq:scalefactorltpmax} gives $\pmax/\scalefactor>1$.

Let $S^\star$ be an optimal solution to a \ColKP{} instance. As in~\eqref{eq:roundset}, for every set $S\subseteq\items$ we have
\begin{equation}
    \scalefactor\hat p(S)\le p(S)
    \le\scalefactor\hat p(S)+|S|\scalefactor
    \le\scalefactor\hat p(S)+n\scalefactor.
    \label{eq:roundset-secondDP}
\end{equation}

We now verify that truncating the profit axis does not remove any state needed to represent $S^\star$. The inner and outer programs build subsets of $S^\star$ in different orders, so we prove the stronger statement that $\hat p(T)\in\mathcal Z$ for every $T\subseteq S^\star$. If all profits are nonnegative, then
\[
    0\le\hat p(T)\le\hat p(S^\star)
      \le\OPT/\scalefactor\le\UB/\scalefactor\le\zmax.
\]
Suppose now that some item profit is negative, and let
$N=\{i\in S^\star:p_i<0\}$. If $N=\emptyset$, the lower bound follows from
$\hat p(T)\ge0$. Otherwise, Lemma~\ref{lem:tight} gives
$|N|\le(n-1)/2$, while $p_i>-\pmax$ for every $i\in N$.
Consequently,
\[
    \hat p(T)\ge\sum_{i\in N}\hat p_i
      >-\frac{n-1}{2}\left(\frac{\pmax}{\scalefactor}+1\right)
      >-\frac{n\pmax}{\scalefactor}\ge\zmin.
\]
For the upper bound, whether or not $N$ is empty,
\[
    \hat p(T)\le\frac{1}{\scalefactor}
       \sum_{i\in S^\star\setminus N}p_i
      =\frac{\OPT+\sum_{i\in N}|p_i|}{\scalefactor}
      <\frac{\UB+n\pmax}{\scalefactor}\le\zmax.
\]

We next show that the restricted program reaches a terminal state whose scaled profit is $\hat p(S^\star)$. Let $S_c^\star=S^\star\cap\items_c$. Every prefix of $S_c^\star$ in the item order of $\items_c$ is a subset of $S^\star$, hence its scaled profit belongs to $\mathcal Z$ and its weight is at most $b$. No transition along the path that builds $S_c^\star$ in the inner program is therefore discarded. The inner table consequently contains a state with cardinality $|S_c^\star|$, scaled profit $\hat p(S_c^\star)$, and weight no greater than
$w(S_c^\star)$.

For $c=0,1,\dots,m$, let $T_c^\star=\bigcup_{j=1}^cS_j^\star$. Since $T_c^\star\subseteq S^\star$, its scaled profit belongs to $\mathcal Z$ and its weight is at most $b$. Induction on $c$ now shows that the outer program retains a state at
\[
    \Bigl(c,\;|T^\star_c|,\;\max_{j\le c}|S^\star_j|,\;
      \hat p(T^\star_c)\Bigr)
\]
of weight no greater than $w(T_c^\star)$. The base case is the initialization at $c=0$. For the inductive step, combine the retained state for $T_{c-1}^\star$ with the inner-table pair for $S_c^\star$. The resulting profit and weight are those of $T_c^\star$, so neither restriction discards the transition.

At $c=m$, the retained state has scaled profit $\hat p(S^\star)$, cardinality $|S^\star|$, dominant cardinality $d(S^\star)$, and weight at most $w(S^\star)\le b$. Since $S^\star$ is color-feasible, it satisfies $2d(S^\star)\le|S^\star|+1$. Thus $\hat p(S^\star)$ is among the terminal profit levels over which the algorithm maximizes. By Proposition~\ref{prop:secondDPcorrect}, every retained terminal state represents an actual subset. Therefore the returned set $\hat S$ is feasible and
\begin{equation}
    \hat p(\hat S)\ge\hat p(S^\star).
    \label{eq:scaled-dominance-secondDP}
\end{equation}
Using~\eqref{eq:roundset-secondDP} and \eqref{eq:scaled-dominance-secondDP}, we conclude that
\begin{align*}
    p(\hat S)
    &\ge\scalefactor\hat p(\hat S) \\
    &\ge\scalefactor\hat p(S^\star) \\
    &>p(S^\star)-|S^\star|\scalefactor \\
    &\ge\OPT-n\scalefactor \\
    &=\OPT-\varepsilon\LB \\
    &\ge(1-\varepsilon)\OPT.
\end{align*}
\end{proof}

\begin{theorem}\label{thm:secondDPtime}
Algorithm~\ref{alg:secondDP-fptas} runs in 
$O(n^5/\varepsilon^2)$ time when all profits are nonnegative and in $O(n^7/\varepsilon^2)$ time for arbitrary integer profits.
\end{theorem}

\begin{proof}
For a nonempty color class $\items_c$, each item layer of the inner dynamic program has $O(|\items_c|\Pi)$ cardinality--profit states. Over the $|\items_c|$ item layers, the inner DP algorithm therefore takes $O(|\items_c|^2\Pi)$ time. Since $\sum_{c\in\colors}|\items_c|^2\le n^2$, the whole inner phase takes $O(n^2\Pi)$ time. Each stage of the outer dynamic program has $O(n^2\Pi)$ states and, from each of them, considers the pairs of $\mathcal S_c$, which are at most $(|\items_c|+1)\Pi=O(|\items_c|\Pi)$. Since every color class is nonempty by assumption, the total running time of the outer DP algorithm is
\[
    \sum_{c\in\colors}O\bigl(n^2\Pi\cdot|\items_c|\Pi\bigr)=O(n^3\Pi^2),
\]
which dominates the inner phase. For nonnegative profits, $\Pi=\lceil3n/\varepsilon\rceil+1=\Theta(n/\varepsilon)$. For arbitrary integer profits, $\Pi=O(n^2/\varepsilon)$ by the same bounds as in Theorem~\ref{thm:time}. Substitution yields the stated running times.

\end{proof}

The difference between the two approaches lies in how they process a color.
The item-by-item DP algorithm adds one item at a time. It must therefore store the counter $a$ for the number of selected items of the current color, but each state has only the two constant-time transitions \emph{pack} and \emph{do not pack}. If the pseudo-polynomial dimension has $q$ levels, this gives the running time $O(n^4 \, q)$.
The color-by-color program avoids the counter $a$ in its outer table because the inner dynamic program has already summarized every possible packing of the next color by its cardinality, profit, and minimum weight. This removes a factor of at most $n$ from the number of outer states. The price is that an outer state can no longer be updated by a constant number of choices: it must be combined with all attainable profit levels for the next color. For fixed cardinalities, this combination is a convolution along the pseudo-polynomial dimension and contributes a factor $q$. The resulting running time is $O(n^3 \, q^2)$. Thus, the color-by-color recursion exchanges a factor $n$ in the state space for a factor $q$ in the transition cost.
For the exact capacity-indexed programs, $q=b+1$, so this exchange is advantageous when the capacity is small relative to $n$. In the FPTAS, however, $q=\Pi$: already for nonnegative profits, $\Pi=\Theta(n/\varepsilon)$, which is larger than the saved factor $n$ by a factor of order $1/\varepsilon$. Accordingly, the running times are $O(n^5/\varepsilon)$ for the item-by-item scheme and $O(n^5/\varepsilon^2)$ for the color-by-color scheme. With arbitrary integer profits, $\Pi=O(n^2/\varepsilon)$, and the corresponding running times become $O(n^6/\varepsilon)$ and $O(n^7/\varepsilon^2)$, respectively. In this case the ratio between the two is $O(n/\varepsilon)$. Hence, the color-by-color recursion, that is useful for small capacities, offers no improvement once the scaled profit range becomes the pseudo-polynomial dimension.

\section{Conclusions}
\label{sec:conclusions}

In this paper, we studied the approximability of the \ColKP, which generalizes the classical Knapsack Problem by adding a color-ordering requirement of the packed items. This requirement can make negative-profit items necessary in an optimal solution and complicates the use of standard knapsack approximation techniques. Our results show that the \ColKP{} nevertheless admits both a constant-factor approximation algorithm and fully polynomial-time approximation schemes, even when item profits are arbitrary integers.

We began by investigating the structure of \ColKP{} optimal solutions. In particular, we identified when negative-profit items can occur and bounded their number, providing the control over signed profits needed for the approximation schemes. We then developed a linear-time approximation algorithm based on the LP relaxation of the \ColKP{} natural ILP formulation, having a worst-case performance ratio of $1/3$.

To obtain the FPTASs, we adapted the two exact dynamic programming algorithms of~\cite{CHM2026} to work with scaled profits. The first scheme, based on a item-by-item recursion, runs in $O(n^5/\varepsilon)$ time for nonnegative profits and $O(n^6/\varepsilon)$ time for arbitrary integer profits. The second, based instead on a color-by-color recursion, runs in $O(n^5/\varepsilon^2)$ and $O(n^7/\varepsilon^2)$ time, respectively. Thus, although both DP algorithms from the literature yield an FPTAS, the item-by-item recursion gives the better worst-case running times under the direct scaling approach considered here.

These results suggest two directions for further research. One is to reduce the running times of the FPTASs, particularly for instances with negative-profit items, where the signed profit range makes the DP algorithms more computationally expensive. A further potential investigation may concern the development of approximation algorithms with a performance guarantee stronger than $1/3$. Nonetheless, this would require to exploit the structure of color-feasible solutions in ways beyond the LP-rounding approach considered here.


\begin{thebibliography}{99}

\bibitem[Balogh et~al.(2013)]{Balogh2013}
J{\'a}nos Balogh, J{\'o}zsef B{\'e}k{\'e}si, Gyorgy Dosa, Hans Kellerer, and Zsolt Tuza. 2013.
\newblock Black and white bin packing.
\newblock In {\em Approximation and Online Algorithms} (Lecture Notes in Computer Science), Thomas Erlebach and Giuseppe Persiano (Eds.).
\newblock Springer, Berlin, Heidelberg, 131--144.


\bibitem[Balogh et~al.(2015)Balogh, B\'ek\'esi, D\'osa, Epstein, Kellerer, Levin and Tuza]{Balogh2015}
J.~Balogh, J.~B\'ek\'esi, G.~D\'osa, L.~Epstein, H.~Kellerer, A.~Levin, Z.~Tuza.
\newblock Offline black and white bin packing.
\newblock \emph{Theoretical Computer Science}, 596:92--101, 2015.
\newblock doi:10.1016/j.tcs.2015.06.045.

\bibitem[Borges et~al.(2024)Borges, Schouery and Miyazawa]{Borges2024}
Y.~G.~F. Borges, R.~C.~S. Schouery, F.~K. Miyazawa.
\newblock Mathematical models and exact algorithms for the colored bin packing
problem.
\newblock \emph{Computers \& Operations Research}, 164:106527, 2024.
\newblock 

\bibitem[Cacchiani et~al.(2022)Cacchiani, Iori, Locatelli and Martello]{Cacchiani2022}
V.~Cacchiani, M.~Iori, A.~Locatelli, S.~Martello.
\newblock Knapsack problems --- an overview of recent advances. Part I: single
knapsack problems.
\newblock \emph{Computers \& Operations Research}, 143:105692, 2022.

\bibitem[Caprara et~al.(2000)Caprara, Kellerer, Pferschy and Pisinger]{Caprara2000}
A.~Caprara, H.~Kellerer, U.~Pferschy, D.~Pisinger.
\newblock Approximation algorithms for knapsack problems with cardinality
constraints.
\newblock \emph{European Journal of Operational Research}, 123(2):333--345,
2000.

\bibitem[Chajakis and Guignard(1994)]{ChajakisGuignard1994}
E.~D. Chajakis, M.~Guignard.
\newblock Exact algorithms for the setup knapsack problem.
\newblock \emph{INFOR: Information Systems and Operational Research},
32(3):124--142, 1994.

\bibitem[Chandra et~al.(1976)Chandra, Hirschberg and Wong]{ChandraHirschbergWong1976}
A.~K. Chandra, D.~S. Hirschberg, C.~K. Wong.
\newblock Approximate algorithms for some generalized knapsack problems.
\newblock \emph{Theoretical Computer Science}, 3(3):293--304, 1976.

\bibitem[Chan(2018)]{Chan2018}
T.~M. Chan.
\newblock Approximation schemes for 0-1 knapsack.
\newblock In \emph{1st Symposium on Simplicity in Algorithms (SOSA 2018)},
5:1--5:12, 2018.

\bibitem[Chen et~al.(2025)Chen, Lian, Mao and Zhang]{Chen2025}
L.~Chen, J.~Lian, Y.~Mao, G.~Zhang.
\newblock A nearly quadratic-time FPTAS for knapsack.
\newblock \emph{SIAM Journal on Computing}, 2025.

\bibitem[Chebil and Khemakhem(2015)]{ChebilKhemakhem2015}
K.~Chebil, M.~Khemakhem.
\newblock A dynamic programming algorithm for the knapsack problem with setup.
\newblock \emph{Computers \& Operations Research}, 64:40--50, 2015.

\bibitem[Ciccarelli et~al.(2026)Ciccarelli, Helber and M\"uhmer]{CHM2026}
F.~Ciccarelli, A.~Helber, E.~M\"uhmer.
\newblock The colored knapsack problem: structural properties and exact
algorithms.
\newblock arXiv:2602.12214, 2026.

\bibitem[Coniglio et~al.(2021)Coniglio, Furini and San Segundo]{Coniglio2021}
S.~Coniglio, F.~Furini, P.~San Segundo.
\newblock A new combinatorial branch-and-bound algorithm for the knapsack
problem with conflicts.
\newblock \emph{European Journal of Operational Research}, 289(2):435--455,
2021.

\bibitem[Deng et~al.(2023)Deng, Jin and Mao]{DengJinMao2023}
M.~Deng, C.~Jin, X.~Mao.
\newblock Approximating knapsack and partition via dense subset sums.
\newblock In \emph{Proceedings of the 2023 ACM-SIAM Symposium on Discrete
Algorithms (SODA 2023)}, 2961--2979, 2023.

\bibitem[D'Ambrosio et~al.(2018)D'Ambrosio, Furini, Monaci and Traversi]{DAmbrosio2018}
C.~D'Ambrosio, F.~Furini, M.~Monaci, E.~Traversi.
\newblock On the product knapsack problem.
\newblock \emph{Optimization Letters}, 12(4):691--712, 2018.

\bibitem[Furini et~al.(2018)Furini, Monaci and Traversi]{FuriniMonaciTraversi2018}
F.~Furini, M.~Monaci, E.~Traversi.
\newblock Exact approaches for the knapsack problem with setups.
\newblock \emph{Computers \& Operations Research}, 90:208--220, 2018.

\bibitem[Ibarra and Kim(1975)]{IbarraKim1975}
O.~H. Ibarra, C.~E. Kim.
\newblock Fast approximation algorithms for the knapsack and sum of subset
problems.
\newblock \emph{Journal of the ACM}, 22(4):463--468, 1975.

\bibitem[Jin(2019)]{Jin2019}
C.~Jin.
\newblock An improved FPTAS for 0-1 knapsack.
\newblock In \emph{46th International Colloquium on Automata, Languages, and
Programming (ICALP 2019)}, 76:1--76:14, 2019.

\bibitem[Kellerer and Pferschy(1999)]{KellererPferschy1999}
H.~Kellerer, U.~Pferschy.
\newblock A new fully polynomial time approximation scheme for the knapsack
problem.
\newblock \emph{Journal of Combinatorial Optimization}, 3(1):59--71, 1999.

\bibitem[Kellerer and Pferschy(2004)]{KellererPferschy2004}
H.~Kellerer, U.~Pferschy.
\newblock Improved dynamic programming in connection with an FPTAS for the
knapsack problem.
\newblock \emph{Journal of Combinatorial Optimization}, 8(1):5--11, 2004.

\bibitem[Kellerer et~al.(2004)Kellerer, Pferschy and Pisinger]{KPP2004}
H.~Kellerer, U.~Pferschy, D.~Pisinger.
\newblock \emph{Knapsack Problems}.
\newblock Springer, Berlin Heidelberg, 2004.

\bibitem[Lawler(1979)]{Lawler1979}
E.~L. Lawler.
\newblock Fast approximation algorithms for knapsack problems.
\newblock \emph{Mathematics of Operations Research}, 4(4):339--356, 1979.

\bibitem[Magazine and Oguz(1981)]{MagazineOguz1981}
M.~J. Magazine, O.~Oguz.
\newblock A fully polynomial approximation algorithm for the 0-1 knapsack
problem.
\newblock \emph{European Journal of Operational Research}, 8(3):270--273, 1981.

\bibitem[Malaguti et~al.(2026)Malaguti, Paronuzzi and Santini]{Malaguti2025}
E.~Malaguti, P.~Paronuzzi, A.~Santini.
\newblock Algorithms and complexity results for the 0--1 knapsack problem with
group fairness.
\newblock \emph{Optimization Letters}, 20:577--594, 2026.
\newblock 

\bibitem[Megiddo(1984)]{Megiddo1984}
Nimrod Megiddo.
\newblock Linear programming in linear time when the dimension is fixed.
\newblock {\emph Journal of the ACM (JACM)}, 31(1):114--127, 1984.

\bibitem[Mucha et~al.(2019)Mucha, W\k{e}grzycki and W{\l}odarczyk]{Mucha2019}
M.~Mucha, K.~W\k{e}grzycki, M.~W{\l}odarczyk.
\newblock A subquadratic approximation scheme for partition.
\newblock In \emph{Proceedings of the 2019 Annual ACM-SIAM Symposium on
Discrete Algorithms (SODA 2019)}, 70--88, 2019.

\bibitem[Nauss(1978)]{Nauss1978}
R.~M. Nauss.
\newblock The 0--1 knapsack problem with multiple-choice constraints.
\newblock \emph{European Journal of Operational Research}, 2(2):125--131,
1978.

\bibitem[Pferschy and Schauer(2009)]{PferschySchauer2009}
U.~Pferschy, J.~Schauer.
\newblock The knapsack problem with conflict graphs.
\newblock \emph{Journal of Graph Algorithms and Applications}, 13(2):233--249,
2009.

\bibitem[Pferschy and Scatamacchia(2018)]{PferschyScatamacchia2018}
U.~Pferschy, R.~Scatamacchia.
\newblock Improved dynamic programming and approximation results for the
knapsack problem with setups.
\newblock \emph{International Transactions in Operational Research},
25(2):667--682, 2018.

\bibitem[Pferschy et~al.(2021)Pferschy, Schauer and Thielen]{PferschySchauerThielen2021}
U.~Pferschy, J.~Schauer, C.~Thielen.
\newblock Approximating the product knapsack problem.
\newblock \emph{Optimization Letters}, 15:2529--2540, 2021.

\bibitem[Sahni(1975)]{Sahni1975}
S.~Sahni.
\newblock Approximate algorithms for the 0/1 knapsack problem.
\newblock \emph{Journal of the ACM}, 22(1):115--124, 1975.

\end{thebibliography}
\end{document}